\documentclass{article}

\def\showauthornotes{1}

\def\showdraftbox{0}

\usepackage{amsmath,amssymb,amsthm,amstext,amsfonts,bbm,algorithm,algorithmicx,graphicx,xspace,nicefrac}
\usepackage{color,stmaryrd,enumerate,latexsym,bm,amsfonts,wrapfig,verbatim,tabularx,textcomp,
  subfig}

\usepackage{comment} 
\usepackage{epsfig} 
\usepackage{nicefrac,bbm}
\usepackage{xspace}
\usepackage{color,fancybox,graphicx,url}
\usepackage{enumitem}
\usepackage{booktabs}
\usepackage{commath}
\usepackage{mdframed}

\usepackage{thm-restate}

\usepackage{fancybox}
\newenvironment{fminipage}%
  {\begin{Sbox}\begin{minipage}}%
  {\end{minipage}\end{Sbox}\fbox{\TheSbox}}

\newtheorem{theorem}{Theorem}[section]

\newtheorem{lemma}[theorem]{Lemma}
\newtheorem{definition}[theorem]{Definition}
\newtheorem{claim}[theorem]{Claim}

\newcommand{\diag}[1]{{\bf Diag}\left({#1}\right)}

\newcommand{\nfrac}[2]{\nicefrac{#1}{#2}} \def\abs#1{\left| #1
  \right|}
\renewcommand{\norm}[1]{\ensuremath{\left\lVert #1
    \right\rVert}}

\newcommand\rea{\mathbb R}

\newcommand{\marginlabel}[1]%
{\mbox{}\marginpar{\it{\raggedleft\hspace{0pt}#1}}}

\newcommand\poly{{\textrm{poly}}}  

\definecolor{Mygray}{gray}{0.8}

\ifcsname ifcommentflag\endcsname\else
\expandafter\let\csname ifcommentflag\expandafter\endcsname
\csname iffalse\endcsname
\fi

\ifnum\showauthornotes=1
\newcommand{\todo}[1]{\colorbox{Mygray}{\color{red}\parbox{\textwidth}{#1}}}
\else
\newcommand{\todo}[1]{}
\fi

\ifnum\showauthornotes=1
\newcommand{\Authornote}[2]{{\sf\small\color{red}{[#1: #2]}}}
\newcommand{\Authoredit}[2]{{\sf\small\color{red}{[#1]}\color{blue}{#2}}}
\newcommand{\Authorcomment}[2]{{\sf \small\color{gray}{[#1: #2]}}}
\newcommand{\Authorfnote}[2]{\footnote{\color{red}{#1: #2}}}
\newcommand{\Authorfixme}[1]{\Authornote{#1}{\textbf{??}}}
\newcommand{\Authormarginmark}[1]{\marginpar{\textcolor{red}{\fbox{
#1:!}}}}
\else
\newcommand{\Authornote}[2]{}
\newcommand{\Authoredit}[2]{}
\newcommand{\Authorcomment}[2]{}
\newcommand{\Authorfnote}[2]{}
\newcommand{\Authorfixme}[1]{}
\newcommand{\Authormarginmark}[1]{}
\fi

\newlength{\pgmtab}  
 {
	\begin{enumerate}}{\end{enumerate}}

\def\qedsketch{\ifmmode\Box\else{\unskip\nobreak\hfil
\penalty50\hskip1em\null\nobreak\hfil$\Box$
\parfillskip=0pt\finalhyphendemerits=0\endgraf}\fi}

\newlength{\tpush}
\newcommand{\handout}[5]{
   \noindent
   \begin{center}
   \framebox{ \vbox{ \hbox to \textwidth { {\bf \coursenum\ :\  \coursename} \hfill #5 }
       \vspace{3mm}
       \hbox to \textwidth { {\Large \hfill #2  \hfill} }
       \vspace{1mm}
       \hbox to \textwidth { {\it #3 \hfill #4} }
     }
   }
   \end{center}
   \vspace*{4mm}
   \newcommand{\lecturenum}{#1}
   \addcontentsline{toc}{chapter}{Lecture #1 -- #2}
}

\ifnum\showdraftbox=1

\else

\fi

\usepackage{algorithm} 
\usepackage[noend]{algpseudocode}
\usepackage[font={small}]{caption}
\usepackage{subfig}
\usepackage{wrapfig}
\newenvironment{tight_enumerate}{
  \begin{enumerate}[leftmargin=2em]
    \setlength{\labelindent}{-1em}
 \setlength{\itemsep}{2pt}
 \setlength{\parskip}{1pt}
}{\end{enumerate}}
\newenvironment{tight_itemize}{
\begin{itemize}[leftmargin=2em]
 \setlength{\itemsep}{2pt}
 \setlength{\parskip}{1pt}
}{\end{itemize}}

\usepackage[title]{appendix}

     \usepackage[nonatbib,final]{neurips_2019}

\usepackage[utf8]{inputenc} 
\usepackage[T1]{fontenc}    
\usepackage{hyperref}       
\usepackage{url}            
\usepackage{booktabs}       
\usepackage{amsfonts}       
\usepackage{nicefrac}       
\usepackage{microtype}      

\newcommand\bb{\boldsymbol{\mathit{b}}}

\newcommand\dd{\boldsymbol{\mathit{d}}}

\renewcommand\gg{\boldsymbol{\mathit{g}}}

\newcommand\rr{\boldsymbol{\mathit{r}}}

\newcommand\vv{\boldsymbol{\mathit{v}}}

\newcommand\xx{\boldsymbol{\mathit{x}}}

\renewcommand\AA{\boldsymbol{\mathit{A}}}
\newcommand\BB{\boldsymbol{\mathit{B}}}
\newcommand\CC{\boldsymbol{\mathit{C}}}

\newcommand\II{\boldsymbol{\mathit{I}}}

\newcommand\RR{\boldsymbol{\mathit{R}}}

\newcommand\WW{\boldsymbol{\mathit{W}}}

\newcommand\xxtilde{\boldsymbol{\widetilde{{x}}}}

\newcommand\Dtil{{\widetilde{{\Delta}}}}

\newcommand\Dopt{{{{\Delta^{\star}}}}}
\newcommand\residual{{{{\gamma}}}}

\newcommand{\Alg}{$p$-\textsc{IRLS}}
\newcommand{\eps}{\varepsilon}
\renewcommand{\epsilon}{\eps}
\newcommand{\etal}{\emph{et al}}
\newcommand{\opt}{\textsc{Opt}}
\renewcommand{\diag}[1]{\ensuremath{\text{diag}\left(#1\right)}}

\newcommand{\sushant}[1]{{\bf \color{blue} Sushant: #1}}
\newcommand{\deeksha}[1]{{\bf \color{blue} Deeksha: #1}}
 \renewcommand{\sushant}[1]{}
 \renewcommand{\deeksha}[1]{}

\title{Fast, Provably convergent IRLS Algorithm for $p$-norm Linear
  Regression \thanks{Code for this work is available at
    \url{https://github.com/utoronto-theory/pIRLS}.} }

\author{%
  Deeksha Adil\\
  Department of Computer Science\\
  University of Toronto\\
  \texttt{deeksha@cs.toronto.edu} \\
  \And
  Richard Peng\\
  School of Computer Science\\
  Georgia Institute of Technology\\
  \texttt{rpeng@cc.gatech.edu}\\
  \And
   Sushant Sachdeva \\
 Department of Computer Science\\
  University of Toronto\\
  \texttt{sachdeva@cs.toronto.edu} 
}

\begin{document}

\maketitle

\begin{abstract}
  Linear regression in $\ell_p$-norm
  is a canonical optimization problem that arises in several
  applications, including sparse recovery, semi-supervised
  learning, and signal processing. 
  Generic convex optimization algorithms for solving
  $\ell_p$-regression are slow in practice. Iteratively Reweighted
  Least Squares (IRLS) is an easy to implement family of algorithms
  for solving these problems that has been studied for over 50 years.
  However, these algorithms often diverge for $p > 3$, and since the
  work of Osborne (1985), it has been an open problem whether there is
  an IRLS algorithm that is guaranteed to converge rapidly for
  $p > 3.$
  We propose {\Alg}, the first IRLS algorithm that provably converges
  geometrically for any $p \in [2,\infty)$. Our algorithm is simple to
  implement and is guaranteed to find a high accuracy solution in a
  sub-linear number of iterations.
  Our experiments demonstrate that it performs even better than our
  theoretical bounds, beats the standard Matlab/CVX implementation for
  solving these problems by 10--50x, and is the fastest among
  available implementations in the high-accuracy regime. 
  %
%
%
%
\end{abstract}



\section{Introduction}

We consider the problem of $\ell_p$-norm linear regression
(henceforth referred to as $\ell_p$-regression),
\begin{equation}
  \label{eq:regression}
  \arg \min_{\xx \in \rea^n} \norm{\AA\xx-\bb}_p,
\end{equation}
where $\AA \in \rea^{m \times n}, \bb \in \rea^{m}$ are given and
$\norm{\vv}_{p} = \left( \sum_i |\vv_i|^p \right)^{\nfrac{1}{p}}$
denotes the $\ell_p$-norm.  This problem generalizes linear regression
and appears in several applications including sparse
recovery~\cite{CandesT05}, low rank matrix
approximation~\cite{chierichetti17a}, and graph based semi-supervised
learning~\cite{AlamgirMU11}.

An important application of $\ell_p$-regression with $p \ge 2$ is
graph based semi-supervised learning (SSL). Regularization using the
standard graph Laplacian (also called a $2$-Laplacian) was introduced
in the seminal paper of Zhu, Gharamani, and Lafferty~\cite{ZhuGL03},
and is a popular approach for graph based SSL, see
e.g.~\cite{ZhouBLWS04, BelkinMN04, ChapelleSZ09, Zhu05}. The
2-Laplacian regularization suffers from degeneracy in the limit of
small amounts of labeled data~\cite{NadlerSZ09}. Several works have
since suggested using the $p$-Laplacian instead~\cite{AlamgirMU11,
  BridleZ13, ZhouB11} with large $p,$ and have established its
consistency and effectiveness for graph based SSL with small amounts
of data~\cite{pmlr-v49-elalaoui16, Calder17, CalderFL19, DejanT17,
  KyngRSS15}.  Recently, $p$-Laplacians have also been used for data
clustering and learning problems \cite{ElmoatazTT15, ElmoatazDT17,
  HafieneFE18}. Minimizing the $p$-Laplacian can be easily seen as an
$\ell_p$-regression problem.

Though $\ell_p$-regression is a convex programming problem, it is very
challenging to solve in practice. General convex programming methods
such as conic programming using interior-point methods (like those
implemented in CVX) are very slow in practice. First order methods do
not perform well for these problems with $p > 2$ since the gradient
vanishes rapidly close to the optimum.

For applications such graph based SSL with $p$-Laplacians, it is
important that we are able to compute a solution $\xx$ that
approximates the optimal solution $\xx^\star$ coordinate-wise rather
than just achieving an approximately optimal objective value, since
these coordinates determine the labels for the vertices. For such
applications, we seek a $(1+\eps)$-approximate solution, an $\xx$ such
that its objective value, $\norm{\AA\xx-\bb}_p^p$ is at most
$(1+\eps)$ times the optimal value $\norm{\AA\xx^\star-\bb}_p^p,$ for
some very small $\eps$ ($10^{-8}$ or so) in order achieve a reasonable
coordinate-wise approximation. Hence, it is very important for the
dependence on $\eps$ be $\log \nfrac{1}{\eps}$ rather than
$\poly(\nfrac{1}{\eps}).$ A $\log \nfrac{1}{\eps}$ guarantee implies a
coordinate-wise convergence guarantee with essentially no loss in the
asymptotic running time. Please see the supplementary material for the
derivation and experimental evaluation of the coordinate-wise
convergence guarantees.

\paragraph{IRLS Algorithms.}
A family of algorithms for solving the $\ell_p$-regression problem are
the IRLS (Iterated Reweighted Least Squares) algorithms. IRLS
algorithms have been discovered multiple times independently and have
been studied extensively for over 50 years e.g.~\cite{Lawson61,
  Rice64, Osborne85, GorodnitskyR97} (see \cite{Burrus12} for a
detailed survey).
%
The main
step in an IRLS algorithm is to solve a weighted least squares
($\ell_2$-regression) problem to compute the next iterate,
\begin{equation}
  \label{eq:IRLS}
\xx^{(t+1)} = \arg\min_{\xx} (\AA\xx-\bb)^{\top}\RR^{(t)}
  (\AA\xx-\bb),
\end{equation}
starting from any initial solution $\xx^{(0)}$ (usually the least
squares solution corresponding to $\RR = \II$). Each iteration can be
implemented by solving a linear system
$ \xx^{(t+1)} \leftarrow
(\AA^{\top}\RR^{(t)}\AA)^{-1}\AA^{\top}\RR^{(t)}\bb.  $ Picking
$\RR^{(t)} = \diag{|\AA\xx^{(t)}-\bb|^{p-2}},$ gives us an IRLS
algorithm where the only fixed point is the minimizer of the
regression problem~\eqref{eq:regression} (which is unique for
$p \in (1,\infty)$).




The basic version of the above IRLS algorithm converges reliably in
practice for $p \in (1.5,3),$ and diverges often even for moderate $p$
(say $p \ge 3.5$~\cite[pg 12]{CalderFL19}). Osborne~\cite{Osborne85}
proved that the above IRLS algorithm converges in the limit for
$p \in [1,3).$ Karlovitz~\cite{Karlovitz70} proved a similar result
for an IRLS algorithm with a line search for even $p>2$. However, both
these results only prove convergence in the limit without any
quantitative bounds, and assume that you start close enough to the
solution. The question of whether a suitable IRLS algorithm converges
geometrically to the optimal solution for \eqref{eq:regression} in a few
iterations has been open for over three decades.


\paragraph{Our Contributions.}
We present {\Alg}, the first IRLS algorithm that provably converges
geometrically to the optimal solution for $\ell_p$-regression for all
$p \in [2, \infty).$ Our algorithm is very similar to the standard
IRLS algorithm for $\ell_p$ regression, and given an $\eps > 0,$
returns a feasible solution $\xx$ for~\eqref{eq:regression} in
$O_p(m^{\frac{p-2}{2(p-1)}} \log \frac{m}{\eps} ) \le O_p(\sqrt{m}
\log \frac{m}{\eps})$ iterations (Theorem~\ref{thm:MainTheorem}). Here
$m$ is the number of rows in $\AA.$ We emphasize that the dependence
on $\eps$ is $\log \frac{1}{\eps}$ rather than
$\poly(\frac{1}{\eps}).$



Our algorithm {\Alg} is very simple to implement, and our experiments
demonstrate that it is much faster than the available implementations
for $p \in (2,\infty)$ in the high accuracy regime. We study its
performance on random dense instances of $\ell_p$-regression, and low
dimensional nearest neighbour graphs for $p$-Laplacian SSL. Our Matlab
implementation on a standard desktop machine runs in at most 2--2.5s
(60--80 iterations) on matrices for size $1000\times 850,$ or graphs
with $1000$ nodes and around 5000 edges, even with $p=50$ and
$\eps=10^{-8}.$ Our algorithm is at least 10--50x faster than the
standard Matlab/CVX solver based on Interior point methods \cite{cvx, gb08}, while
finding a better solution. We also converge much faster than
IRLS-homotopy based algorithms~\cite{CalderFL19} that are not even
guaranteed to converge to a good solution. For larger $p,$ say
$p > 20,$ this difference is even more dramatic, with {\Alg} obtaining
solutions with at least 4 orders of magnitude smaller error with the
same number of iterations. Our experiments also indicate that {\Alg}
scales much better than as indicated by our theoretical bounds, with
the iteration count almost unchanged with problem size, and growing
very slowly (at most linearly) with $p.$


\subsection{Related Works and Comparison}
IRLS algorithms have been used widely for various problems due to
their exceptional simplicity and ease of implementation, including
compressive sensing~\cite{ChartrandY08}, sparse signal
reconstruction~\cite{GorodnitskyR97}, and Chebyshev approximation in
FIR filter design \cite{BarretoB94}.  There have been various attempts
at analyzing variants of IRLS algorithm for $\ell_p$-norm
minimization.  We point the reader to the survey by
Burrus~\cite{Burrus12} for numerous pointers and a thorough history.

The works of Osborne~\cite{Osborne85} and Karlovitz~\cite{Karlovitz70}
mentioned above only prove convergence in the limit without
quantitative bounds and under assumptions on $p$ and that we start close
enough. Several works show that it is similar to Newton's method
(e.g.~\cite{Kahng72, BurrusBS94}), or that adaptive step sizes help
(e.g.~\cite{BurrusV99, BurrusV12}) but do not prove any guarantees.


A few notable works prove convergence guarantees for IRLS algorithms
for sparse recovery (even $p < 1$ in some
cases)~\cite{DaubechiesDFG08, DaubechiesDFG10, LN18}, and for low-rank
matrix recovery~\cite{FornasierRW11}.
Quantitative convergence bounds for IRLS algorithms for $\ell_1$ are
given by Straszak and Vishnoi~\cite{StraszakV16Flow, StraszakV16,
  StraszakV16IRLS}, inspired by slime-mold dynamics. Ene and Vladu
give IRLS algorithms for $\ell_1$ and $\ell_\infty$~\cite{EneV19}.
However, both these works have $\poly(\nfrac{1}{\eps})$ dependence in
the number of iterations, with the best result by~\cite{EneV19}
having a total iteration count roughly $m^{\nfrac{1}{3}} \eps^{-\nfrac{2}{3}}$.




 The most relevant theoretical results for $\ell_p$-norm minimization are
Interior point methods~\cite{NesterovN94}, the homotopy method of
Bubeck~\etal~\cite{BubeckCLL18}, and the iterative-refinement method
of Adil~\etal.~\cite{AdilKPS19}.  The convergence bounds we prove on
the number of iterations required by {\Alg}
(roughly~$m^{\frac{p-2}{2p-2}}$) has a better dependence on $m$ than
Interior Point methods (roughly $m^{\nfrac{1}{2}}$), but marginally
worse than the dependence in the work of Bubeck
\etal.~\cite{BubeckCLL18} (roughly $m^{\frac{p-2}{2p}}$) and
Adil~\etal.~\cite{AdilKPS19} (roughly $m^{\frac{p-2}{2p +
    (p-2)}}$). Note that we are comparing the dominant polynomial
terms, and ignoring the smaller $\poly(p\log \nfrac{m}{\eps})$
factors. A follow-up work in this line by a subset of the
authors~\cite{AdilS19} focuses on the large $p$ case, achieving a
similar running time to~\cite{AdilKPS19}, but with linear dependence
on $p.$
Also related, but not directly comparable are the works of
Bullins~\cite{Bullins18} (restricted to $p=4$) and the work of
Maddison~\etal.~\cite{Maddison18} (first order method with a
dependence on the condition number, which could be large).

More importantly, in contrast with comparable second order
methods~\cite{BubeckCLL18, AdilKPS19}, our algorithm is far simpler to
implement, and has a \emph{locally greedy} structure that allows for
greedily optimizing the objective using a line search, resulting in
much better performance in practice than that guaranteed by our
theoretical bounds. Unfortunately, there are also no available
implementations for any of the above discussed methods (other than
interior point methods) in order to make a comparison.

Another line of heuristic algorithms combines IRLS algorithms with a
homotopy based approach
(e.g.~\cite{Kahng72}. See~\cite{Burrus12}). These methods start from a
solution for $p=2,$ and slowly increase $p$ multiplicatively, using an
IRLS algorithm for each phase and the previous solution as a starting
point. These algorithms perform better in practice than usual IRLS
algorithms. However, to the best of our knowledge, they are not
guaranteed to converge, and no bounds on their performance are
known. Rios~\cite{Flores19} provides an efficient implementation of
such a method based on the work of Rios~\etal.~\cite{CalderFL19},
along with detailed experiments. Our experiments show that our
algorithm converges much faster
than 
the implementation from Rios 
(see Section~\ref{sec:experiments}).

\section{Preliminaries}


We first define some terms that we will use in the formal analysis of
our algorithm. For our analysis we use a more general form of the
$\ell_p$-regression problem,
\begin{equation}
\label{eq:ConstrainedReg}
\arg \min_{\xx:\CC\xx = \dd} \norm{\AA\xx-\bb}_p.
\end{equation}
Setting  $\CC$ and $\dd$ to be empty recovers the standard
$\ell_p$-regression problem.
%
\begin{definition}[Residual Problem]
\label{def:residual}
The residual problem of \eqref{eq:ConstrainedReg} at $\xx$ is defined as,
\begin{align*}
\max_{\Delta: \CC\Delta = 0}  \quad & \gg^{\top}\AA\Delta - 2p^2\Delta^{\top}\AA^{\top}\RR\AA\Delta - p^p \|\AA\Delta\|_p^p.
\end{align*}
Here $\RR = \diag{|\AA\xx-\bb|^{p-2}}$ and $\gg = p\RR(\AA\xx-\bb)$ is the gradient of the objective at $\xx$. Define $\residual(\Delta)$ to denote the objective of the residual problem evaluated at $\Delta$.
\end{definition}

\begin{definition}[Approximation to the Residual Problem] Let
  $\kappa \geq 1$ and $\Dopt$ be the optimum of the residual
  problem. A $\kappa$-approximate solution to the residual problem is
  $\Dtil$ such that  $\CC\Dtil = 0,$ and
 $\residual(\Dtil)\geq \frac{1}{\kappa} \residual(\Dopt).$
\end{definition}


\section{Algorithm and Analysis}
\label{sec:Algo}

\begin{algorithm}
\caption{{\Alg} Algorithm}
\label{alg:MainAlgo}
 \begin{algorithmic}[1]
 \Procedure{{\Alg}}{$\AA, \bb, \epsilon, \CC,\dd$}
 \State $\xx \leftarrow  \arg\min_{\CC\xx = \dd} \norm{\AA\xx-\bb}^2_2.$
 \State $i  \leftarrow \|\AA\xx - \bb\|_p^p/16p$
 \While{$\frac{\epsilon}{16p(1+\epsilon)} \norm{\AA\xx-\bb}_p^p < i$}
 \State $\RR \leftarrow |\AA\xx -\bb |^{p-2}$ \label{alg:line:DefineR}
 \State $\gg = p\RR(\AA\xx-\bb)$ \label{alg:line:DefineG}
 \State $s \leftarrow \frac{1}{2} i^{(p-2)/p} m^{-(p-2)/p}$ \label{alg:line:DefineS}
 \State $\Dtil \leftarrow \arg \min_{\gg^{\top}\AA\Delta = i/2, \CC\Delta = 0}  \quad \Delta^{\top}\AA^{\top}(\RR + s\II)\AA\Delta$ \label{alg:line:l2}
  \State $\alpha \leftarrow$ {\sc LineSearch}$(\AA,\bb,\xx^{(t)},\Dtil)$  \Comment{$\alpha = \arg \min_{\alpha}\|\AA(\xx-\alpha\Dtil)-\bb\|_p^p$}\label{alg:line:DefineScale}
 \State $\xx^{(t+1)} \leftarrow \xx^{(t)} - \alpha\Dtil$
  \If{ {\sc InsufficientProgressCheck($\AA,\RR+s\II,\Dtil, i $)}}
 \ $i \leftarrow i/2$
 \EndIf
\EndWhile
\State \Return $\xx$
 \EndProcedure 
 \end{algorithmic}
\end{algorithm}

\begin{algorithm}
\caption{Check Progress}
\label{alg:CheckProgress}
\begin{algorithmic}[1]
\Procedure{InsufficientProgressCheck}{$\AA, \RR, \Delta,i$}
\State $\lambda \leftarrow 16p$
\State $k \leftarrow \frac{p^p\norm{\AA\Delta}_p^p}{2p^2\Delta^{\top}\AA^{\top}\RR \AA \Delta}$
\State $\alpha_0 \leftarrow \min\left\{\frac{1}{16\lambda}, \frac{1}{(16\lambda k)^{1/(p-1)}}\right\}$\label{alg:line:DefineAlpha}
\If {$\residual(\alpha_0 \cdot \Dtil) < \frac{\alpha_0}{4} i$ or
  $\Delta^{\top}\AA^{\top}(\RR+s\II) \AA \Delta > \lambda i/p^2$} \Return true
\Else \ \Return false
\EndIf
\EndProcedure 
\end{algorithmic}
\end{algorithm}

Our algorithm {\Alg}, described in Algorithm~\eqref{alg:MainAlgo}, is
the standard IRLS algorithm (equation~\ref{eq:IRLS}) with few key
modifications.  The first difference is that at each iteration $t$, we
add a small systematic padding $s^{(t)} \II$ to the weights
$\RR^{(t)}.$ The second difference is that the next iterate
$\xx^{(t+1)}$ is calculated by performing a line search along the line
joining the current iterate $\xx^{(t)}$ and the standard IRLS iterate
$\xxtilde^{(t+1)}$ at iteration $t+1$ (with the modified weights)
\footnote{Note that {\Alg} has been written in a slightly different
  but equivalent formulation, where it solves for
  $\Dtil = \xx^{(t)} - \xxtilde^{(t+1)}.$}.  Both these modifications
have been tried in practice, but primarily from practical
justifications: padding the weights avoids ill-conditioned matrices,
and line-search can only help us converge faster and improves
stability~\cite{Karlovitz70, BurrusV99, BurrusV99}. Our key
contribution is to show that these modifications together allow us to
provably make $\Omega_p(m^{-\frac{p-2}{2(p-1)}})$ progress towards the
optimum, resulting in a final iteration count of
$O_p(m^{\frac{p-2}{2(p-1)}} \log \frac{m}{\eps}).$ Finally, at every
iteration we check if the objective value decreases sufficiently, and
this allows us to adjust $s^{(t)}$ appropriately.  We emphasize here
that our algorithm always converges. 
We prove the following theorem:
%
%
%
\begin{theorem}
  \label{thm:MainTheorem}
  Given any
  $\AA \in \rea^{m \times n}, \bb \in \rea^{m}, \epsilon>0, p \geq 2$
  and
  $\xx^{\star} = \arg \min_{\xx: \CC\xx=\dd}
  \norm{\AA\xx-\bb}_p^p$. Algorithm \ref{alg:MainAlgo} returns $\xx$
  such that
  $\norm{\AA\xx- \bb}_p^p \leq (1+\epsilon) \norm{\AA\xx^{\star}-
    \bb}_p^p$ and $\CC\xx = \dd$, in at most
  $O\left(p^{3.5}m^{\frac{p-2}{2(p-1)}}\log\left(\frac{m}{\epsilon}\right)\right)$
  iterations.
\end{theorem}

The approximation guarantee on the objective value can be translated to a guarantee on coordinate wise convergence. For details on this refer to the supplementary material.

\subsection{Convergence Analysis}
The analysis, at a high level, is based on iterative refinement
techniques for $\ell_p$-norms developed in the work of
Adil~\etal~\cite{AdilKPS19} and Kyng~\etal~\cite{KyngPSW19}.  These
techniques allow us to use a crude $\kappa$-approximate solver
for the {\it residual} problem (Definition~\ref{def:residual})
$O_p(\kappa \log \frac{m}{\eps})$ number of times to obtain a
$(1+\eps)$ approximate solution for the $\ell_p$-regression problem
(Lemma \ref{lem:IterativeRefinement}). 

In our algorithm, if we had solved the standard weighted $\ell_2$ problem instead, $\kappa$ would be unbounded. The padding added to the weights allow us to prove that the solution to weighted $\ell_2$ problem gives a bounded approximation to the residual problem provided we have the correct padding, or in other words correct value of $i$ (Lemma \ref{lem:Approximation}).
We will show that the number of iterations where we are adjusting the
value of $i$ are small. Finally, Lemma \ref{lem:Termination} shows that when the algorithm
terminates, we have an $\eps$-approximate solution to our main
problem. The remaining lemma of this section, Lemma
\ref{lem:Invariant} gives the loop invariant which is used at several
places in the proof of Theorem \ref{thm:MainTheorem}. Due to space
constraints, we only state the main lemmas here and defer the proofs
to the supplementary material.

We begin with the lemma that talks about our overall iterative
refinement scheme. The iterative refinement scheme in \cite{AdilKPS19} and \cite{KyngPSW19} has an exponential
dependence on $p$. We improve this dependence to a small polynomial in $p$.
\begin{restatable}{lemma}{IterativeRefinement}(Iterative Refinement). \label{lem:IterativeRefinement}
Let $p \geq 2$, and $\kappa \geq 1$. Starting from $\xx^{(0)} = \arg\min_{\CC\xx =\dd} \norm{\AA\xx-\bb}_2^2$,  and iterating as, $\xx^{(t+1)} = \xx^{(t)} - \Delta$, where $\Delta$ is a $\kappa$-approximate solution to the residual problem (Definition \ref{def:residual}), we get an $\epsilon$-approximate solution to \eqref{eq:ConstrainedReg} in at most $O\left(p^2 \kappa \log \left(\frac{m}{\epsilon} \right)\right)$ calls to a $\kappa$-approximate solver for the residual problem.
\end{restatable}
The next lemma talks about bounding the approximation factor $\kappa$, when we have the right value of $i$.


\begin{restatable}{lemma}{Approximation}(Approximation).\label{lem:Approximation}
Let $\RR,\gg,s,\alpha$ be as defined in lines \eqref{alg:line:DefineR}, \eqref{alg:line:DefineG}, \eqref{alg:line:DefineS} and  \eqref{alg:line:DefineScale} of Algorithm \ref{alg:MainAlgo}. Let $\alpha_0$ be as defined in line \eqref{alg:line:DefineAlpha} of Algorithm \ref{alg:CheckProgress} and $\Dtil$ be the solution of the following program,
\begin{align}
\label{eq:ProblemSolve}
\arg \min_{\Delta} \Delta^{\top}\AA^{\top}(\RR +
                     s\II)\AA\Delta
\quad \text{s.t.} \quad \gg^{\top}\AA\Delta = i/2, \CC\Delta = 0.
\end{align}
If $\Dtil^{\top}\AA^{\top}(\RR +s\II)\AA\Dtil \leq \lambda i/p^2 $ and $\residual(\alpha_0 \cdot \Dtil) \geq \frac{\alpha_0 i}{4}$, then $\alpha \cdot \Dtil$ is an $O\left(p^{1.5}m^{\frac{p-2}{2(p-1)}}\right)$- approximate solution to the residual problem.
\end{restatable}


We next present the loop invariant followed by the conditions for the termination.
\begin{restatable}{lemma}{Invariant}(Invariant)
\label{lem:Invariant}
At every iteration of the while loop, we have  $\CC\xx^{(t)} = \dd,$ 
$ \frac{(\|\AA\xx^{(t)}-\bb\|_p^p -  \|\AA\xx^{\star}-b\|_p^p)}{16p} \leq i$ and $i \geq \frac{\epsilon}{16p(1+\epsilon)}\|\AA\xx^{(0)} - \bb\|_p^p m^{-(p-2)/2}$.
\end{restatable}

\begin{restatable}{lemma}{termination}(Termination).
\label{lem:Termination}
Let $i$ be such that $(\|\AA\xx^{(t)}-\bb\|_p^p -  \|\AA\xx^{\star}-b\|_p^p)/16p \in (i/2,i]$. Then,
\[ \textstyle i \leq \frac{\epsilon}{16p(1+\epsilon)}
  \|\AA\xx^{(t)}-\bb\|_p^p \Rightarrow \|\AA\xx^{(t)}-\bb\|_p^p \leq
  (1+\epsilon)\opt.
\]
and,
\[ \textstyle \|\AA\xx^{(t)}-\bb\|_p^p \leq (1+\epsilon)\opt
  \Rightarrow i \leq 2\frac{\epsilon}{16p(1+\epsilon)}
  \|\AA\xx^{(t)}-\bb\|_p^p.
\]
\end{restatable}

We next see how Lemmas \ref{lem:IterativeRefinement},\ref{lem:Approximation}, \ref{lem:Invariant}, and \ref{lem:Termination} together imply our main result, Theorem \ref{thm:MainTheorem}.

\subsection{Proof of Theorem \ref{thm:MainTheorem}}
\begin{proof}
We first show that at termination, the algorithm returns an $\epsilon$-approximate solution. We begin by noting that the quantity $i$ can only decrease with every iteration. At iteration $t$, let $i_0$ denote the smallest number such that $(\|\AA\xx^{(t)}-\bb\|_p^p -  \|\AA\xx^{\star}-b\|_p^p)/16p \in (i_0/2,i_0]$. Note that $i$ must be at least $i_0$ (Lemma  \ref{lem:Invariant}). Let us first consider the termination condition of the while loop. When we terminate, $\frac{\epsilon \|\AA\xx^{(t)}-\bb\|_p^p}{16p(1+\epsilon)} \geq i \geq i_0$. Lemma \ref{lem:Termination} now implies that $\norm{\AA\xx^{(t)}-\bb}_p^p \leq (1+\epsilon)OPT$. Lemma \ref{lem:Invariant} also shows that at each iteration our solution satisfies $\CC\xx^{(t)} = \dd$, therefore the solution returned at termination also satisfies the subspace constraints.

We next prove the running time bound. Note that the objective is non
increasing with every iteration. This is because the {\sc LineSearch}
returns a factor that minimizes the objective given a direction
$\Dtil$, i.e.,
$\alpha = \arg\min_{\delta}\|\AA(\xx-\delta\Dtil)-\bb\|_p^p$, which
could also be zero.

We now show that at every iteration the algorithm either reduces $i$
or finds $\Dtil$ that gives a
$O\left( p^{1.5}m^{\frac{p-2}{2(p-1)}}\right)$-approximate solution to
the residual problem. Consider an iteration where the
  algorithm does not reduce $i.$ It suffices to prove that in this
  iteration, the algorithm obtains an
  $O\left( p^{1.5}m^{\frac{p-2}{2(p-1)}}\right)$-approximate solution
  to the residual problem. Since the algorithm does not reduce $i,$ we
  must have $\residual(\alpha_0 \Dtil) \geq \alpha_0 i /4,$ and
  $\Dtil^{\top}\AA^{\top}(\RR+s\II) \AA \Dtil \leq \lambda i/p^2.$ It follows from Lemma~\ref{lem:Approximation}, we know that $\Dtil$ gives the required approximation to the residual problem.

Thus, the algorithm either reduces $i$ or returns an
$O\left( p^{1.5}m^{\frac{p-2}{2(p-1)}}\right)$-approximate solution to
the residual problem. The number of steps in which we reduce $i$ is at
most
$\log(i_{initial}/i_{min}) = p \log\left(\frac{m}{\epsilon}\right)$
(Lemma \ref{lem:Invariant} gives the value of $i_{min}$).  By Lemma
\ref{lem:IterativeRefinement}, the number of steps where the algorithm
finds an approximate solution before it has found a
$(1+\eps)$-approximate solution is at most
$O\left(p^{3.5}m^{\frac{p-2}{2(p-1)}} \log \left(\frac{m}{\epsilon}
  \right)\right)$. Thus, the total number of iterations required by
our algorithm is
$O\left(p^{3.5}m^{\frac{p-2}{2(p-1)}} \log \left(\frac{m}{\epsilon}
  \right)\right)$, completing the proof of the theorem.
\end{proof}

\section{Experiments}
\label{sec:experiments}
\begin{wrapfigure}{r}{0.33\textwidth}
  \centering
  \vspace{-20pt}
    \includegraphics[width=0.3\textwidth]{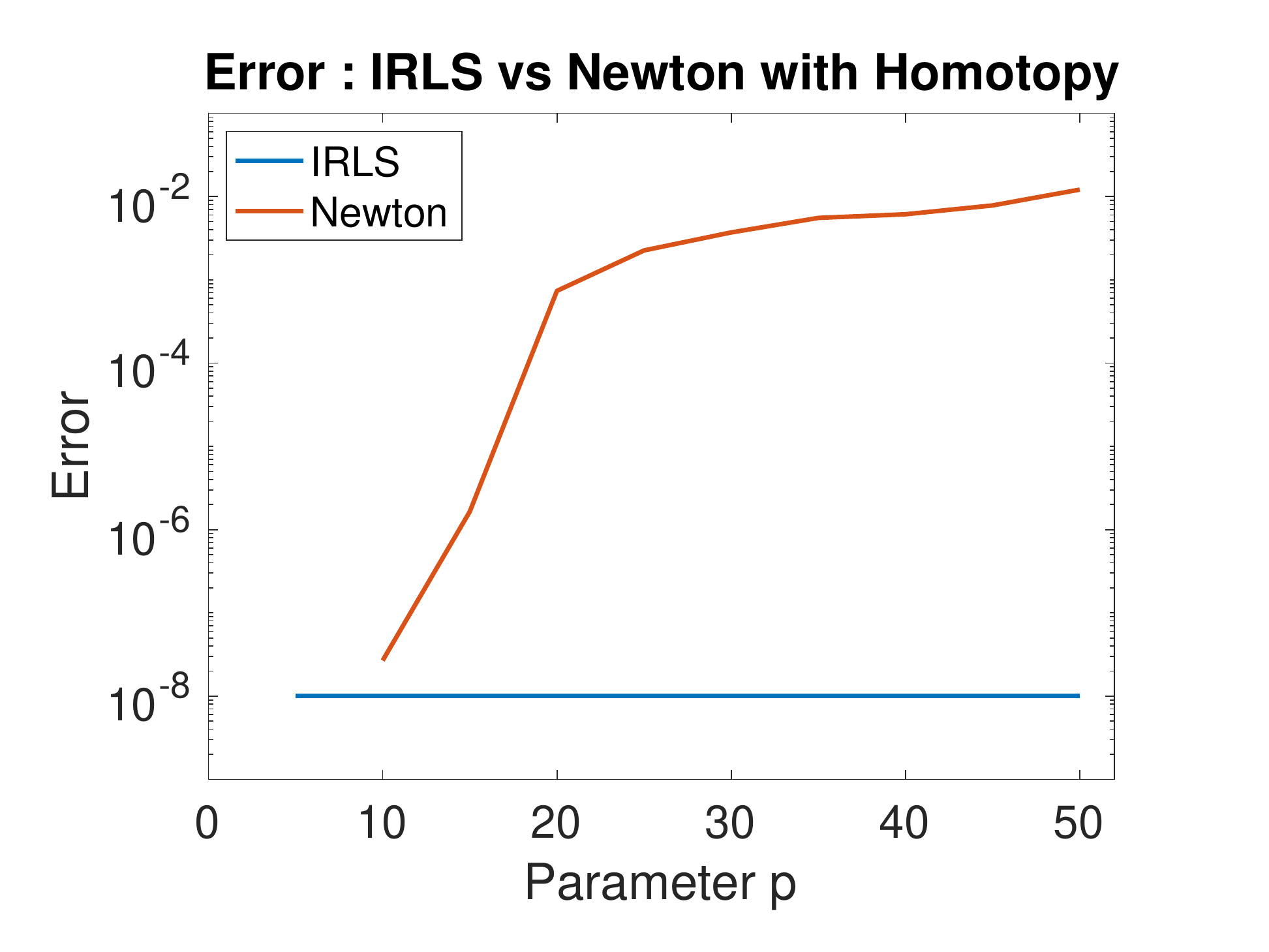}
\caption{Averaged over 100 random samples. Graph: $1000$ nodes ($5000$-$6000$ edges). Solver: PCG with Cholesky preconditioner.}
\label{fig:IRLSvsNt}
\end{wrapfigure}

In this section, we detail our results from experiments studying the
performance of our algorithm, {\Alg}. We implemented our algorithm in
Matlab on a standard desktop machine, and evaluated its performance on
two types of instances, random instances for $\ell_p$-regression, and
graphs for $p$-Laplacian minimization. We study the scaling behavior
of our algorithm as we change $p,\epsilon,$ and the size of the
problem. We compare our performance to the Matlab/CVX solver that is
guaranteed to find a good solution, and to the IRLS/homotopy based
implementation from~\cite{CalderFL19} that is not guaranteed to
converge, but runs quite well in practice. We now describe our
instances, parameters and experiments in detail.

\begin{figure}
\subfloat[Size of $\AA$ fixed to $1000 \times 850$.]{  \includegraphics[width =0.3\textwidth]{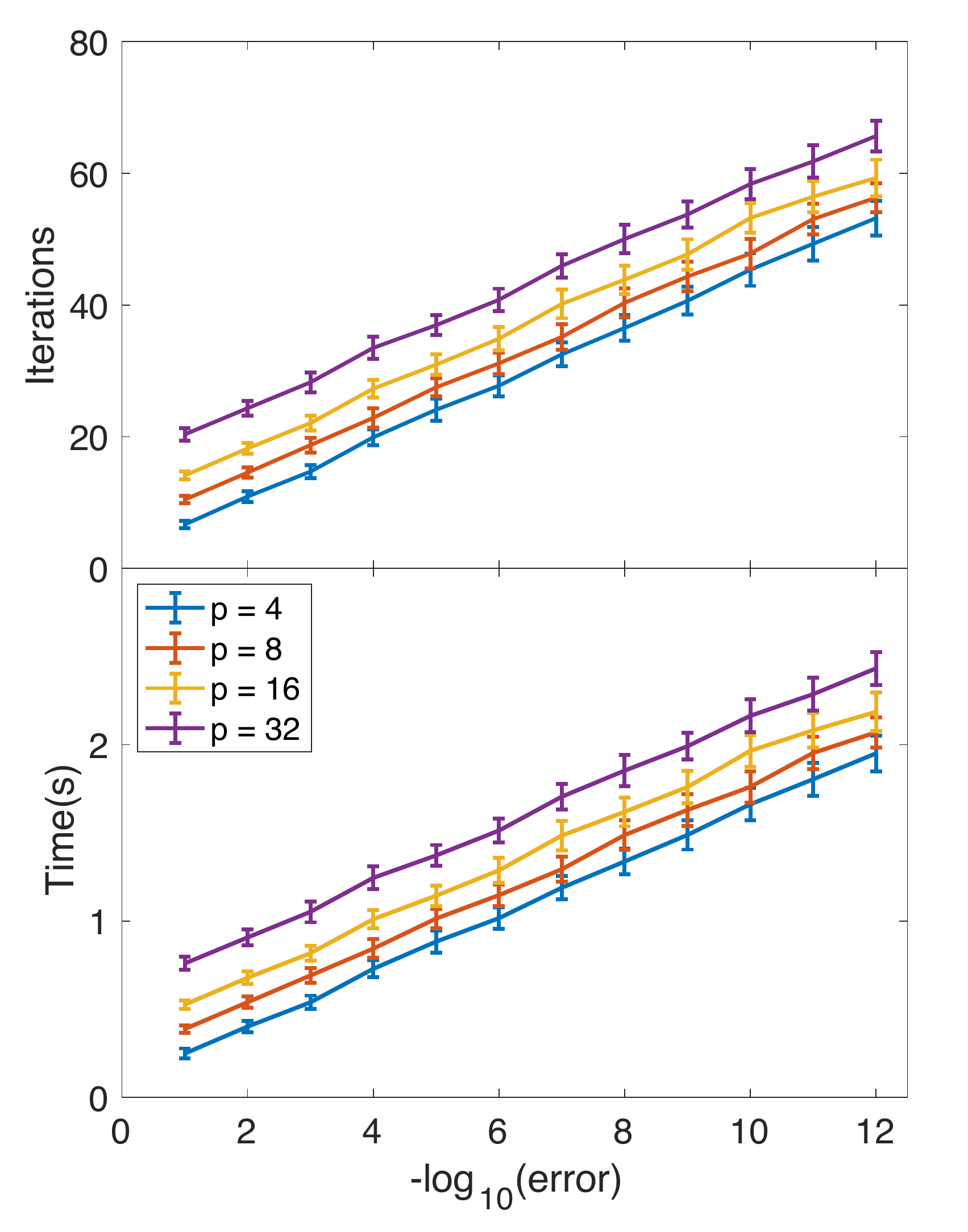}}
\hfill
\subfloat[Sizes of $\AA$: $(50+100(k-1))\times 100k$. Error $\eps = 10^{-8}.$]{  \includegraphics[width =0.3\textwidth]{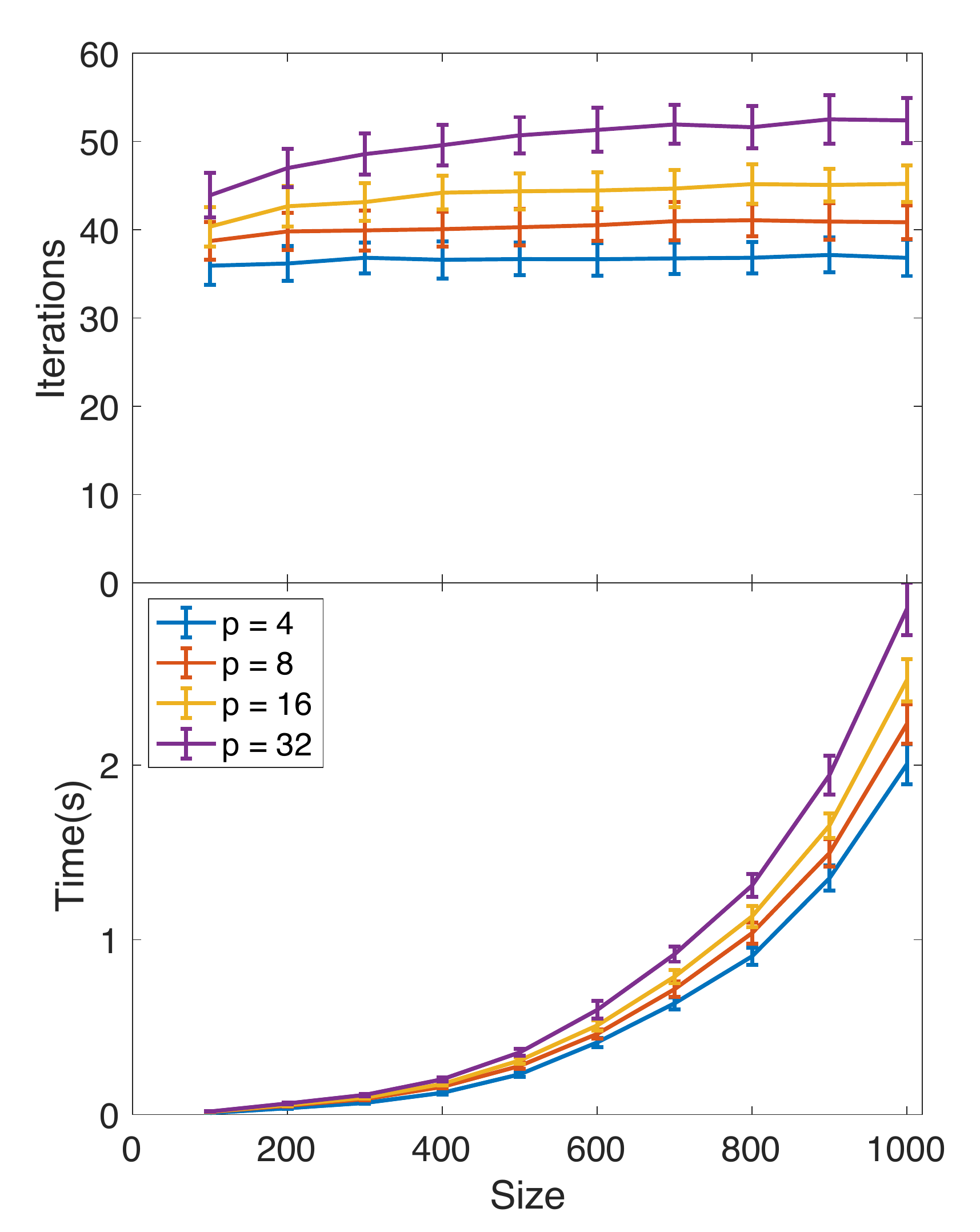}}
\hfill
 \subfloat[Size of $\AA$ is fixed to $1000 \times 850$. Error $\eps = 10^{-8}.$]{  \includegraphics[width =0.3\textwidth]{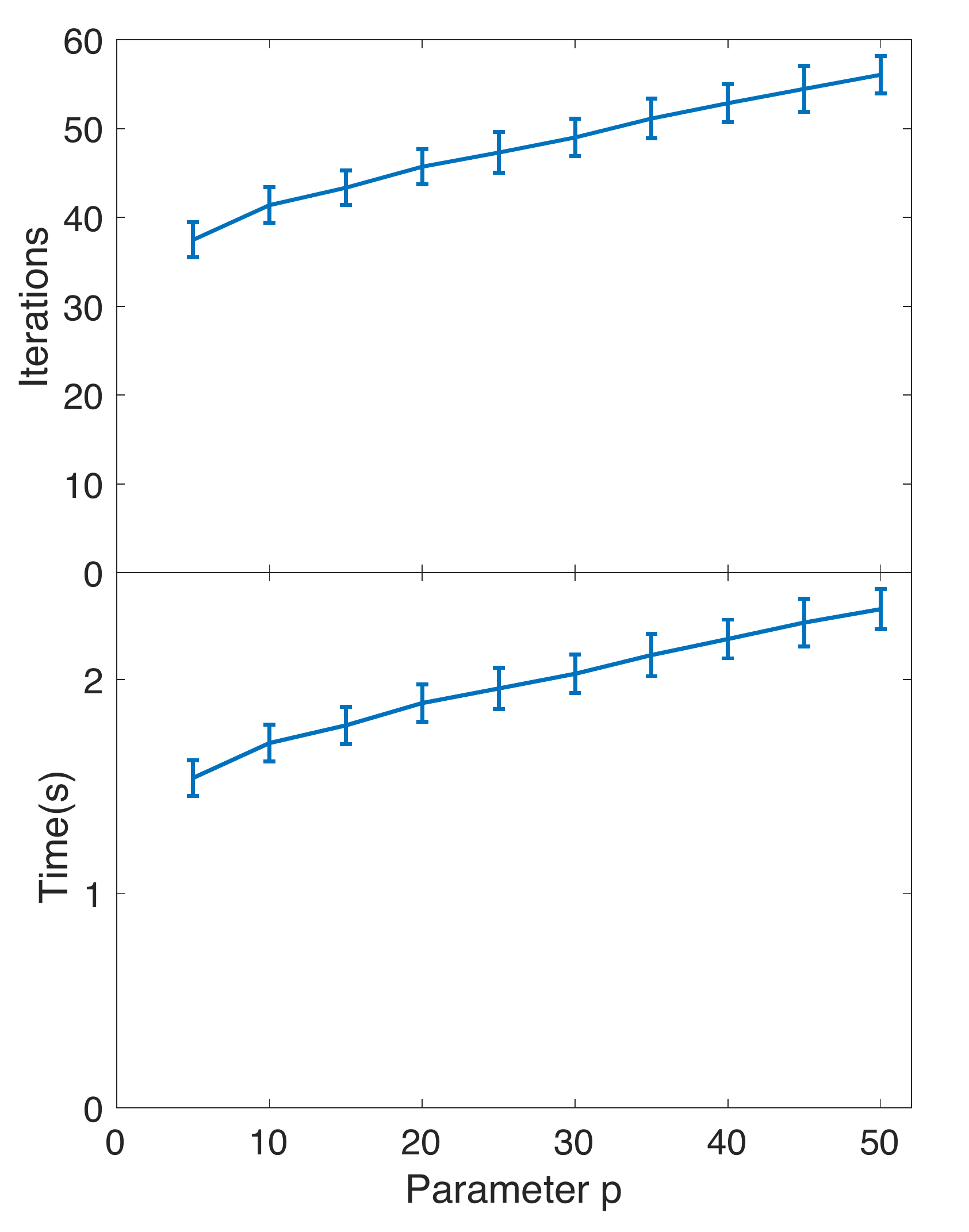}}
   \caption{Random Matrix instances. Comparing the number of iterations and time taken by our algorithm with the parameters. Averaged over 100 random samples for $\AA$ and $\bb$. Linear solver used : backslash.}
   \label{fig:Matrices}
\end{figure}

\begin{figure}
\subfloat[Size of graph fixed to $1000$ nodes (around $5000$-$6000$ edges).]{  \includegraphics[width =0.3\textwidth]{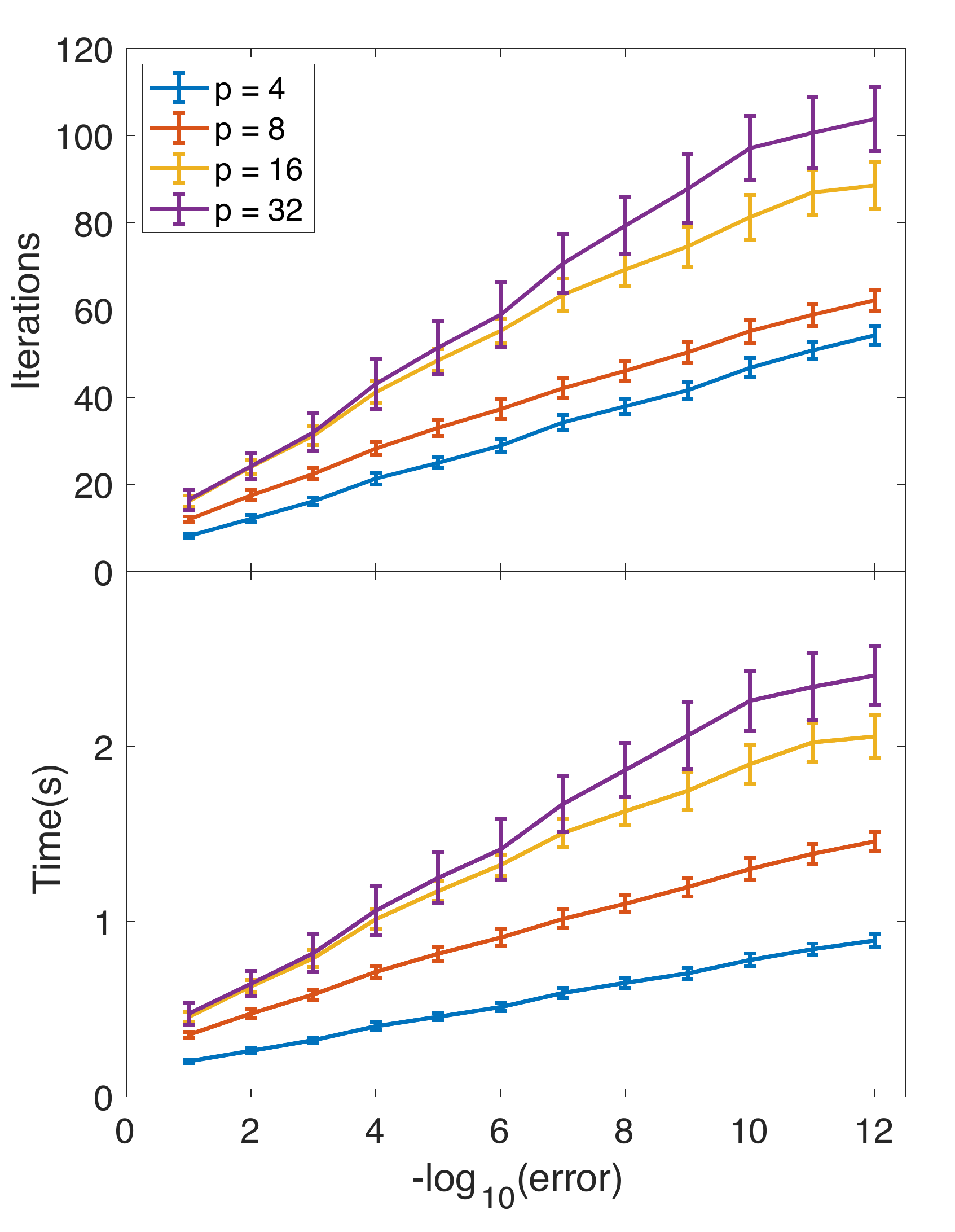}}
\hfill
\subfloat[Number of nodes: $100k$. Error $\eps = 10^{-8}.$]{  \includegraphics[width =0.3\textwidth]{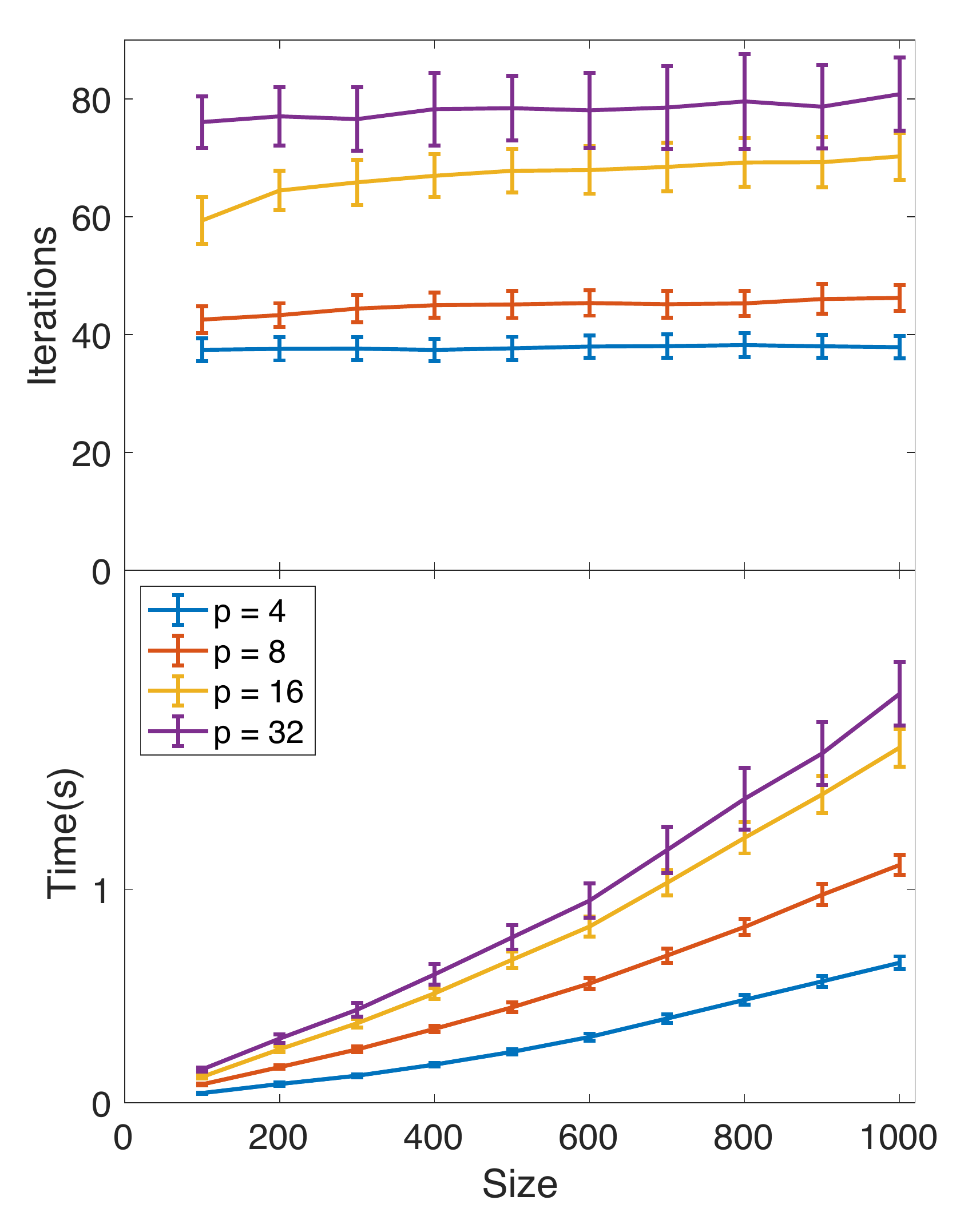}}
\hfill
 \subfloat[Size of graph fixed to $1000$ nodes (around $5000$-$6000$ edges). Error $\eps = 10^{-8}.$]{  \includegraphics[width =0.3\textwidth]{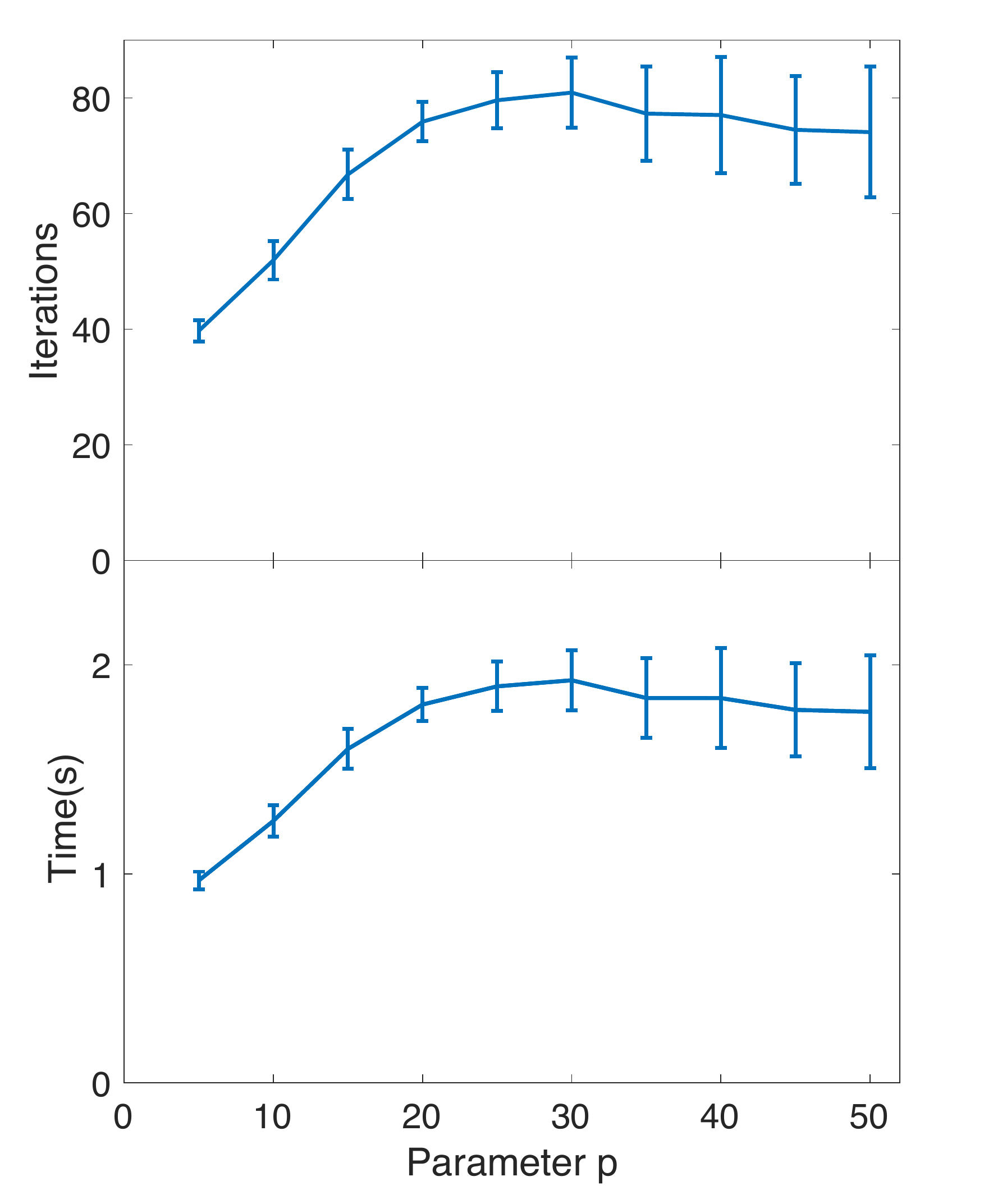}}
   \caption{Graph Instances. Comparing the number of iterations and time taken by our algorithm with the parameters. Averaged over 100 graph samples. Linear solver used : backslash.}
   \label{fig:Graphs}
\end{figure}

\begin{figure}
\centering
   \subfloat[Fixed $p = 8$. Size of matrices: $100k \times (50+100(k-1))$.] {\includegraphics[width=0.23\textwidth]{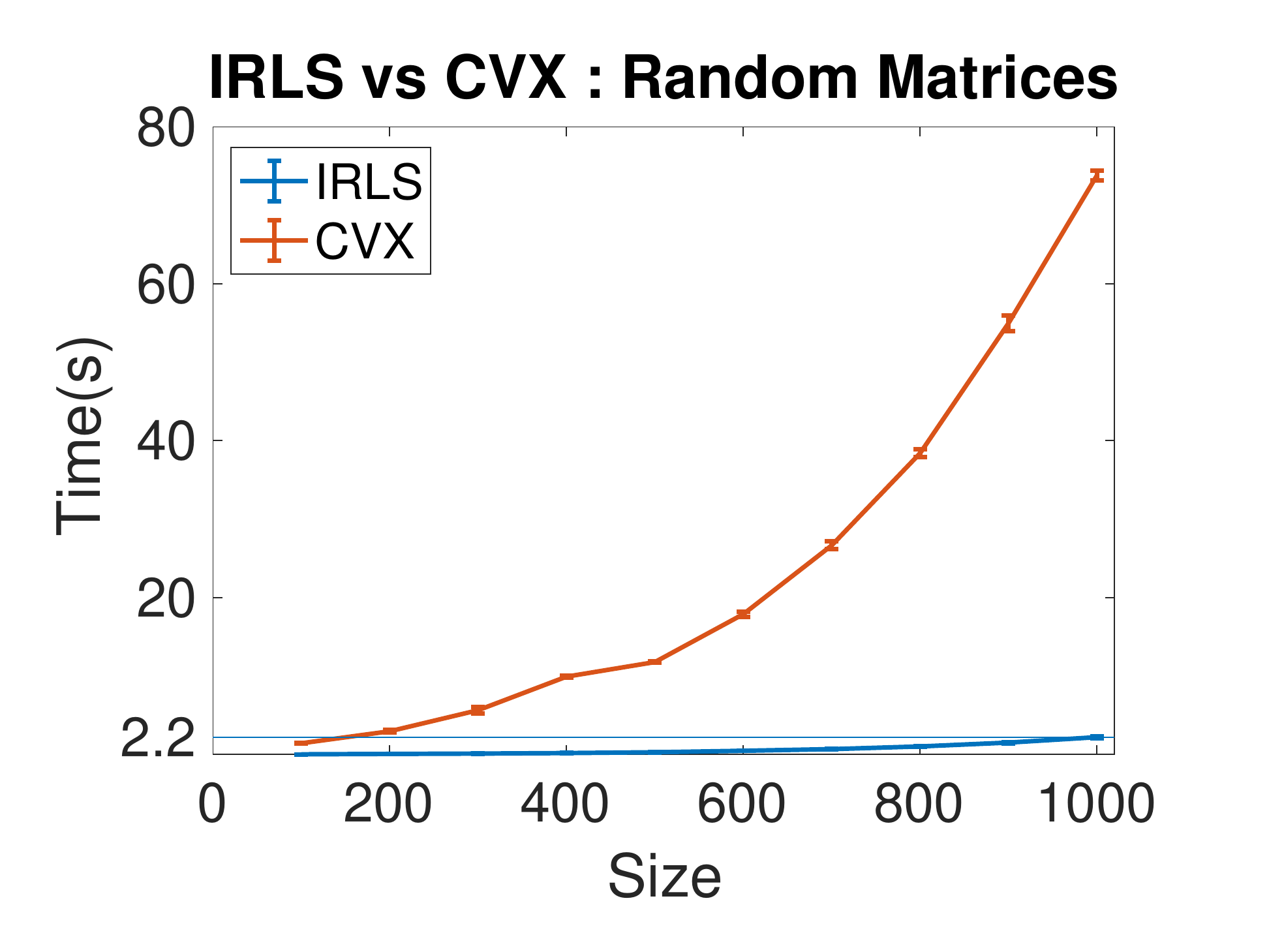}}
  \hfill
  \subfloat[Size of matrices fixed to $500 \times 450$.]{\includegraphics[width=0.23\textwidth]{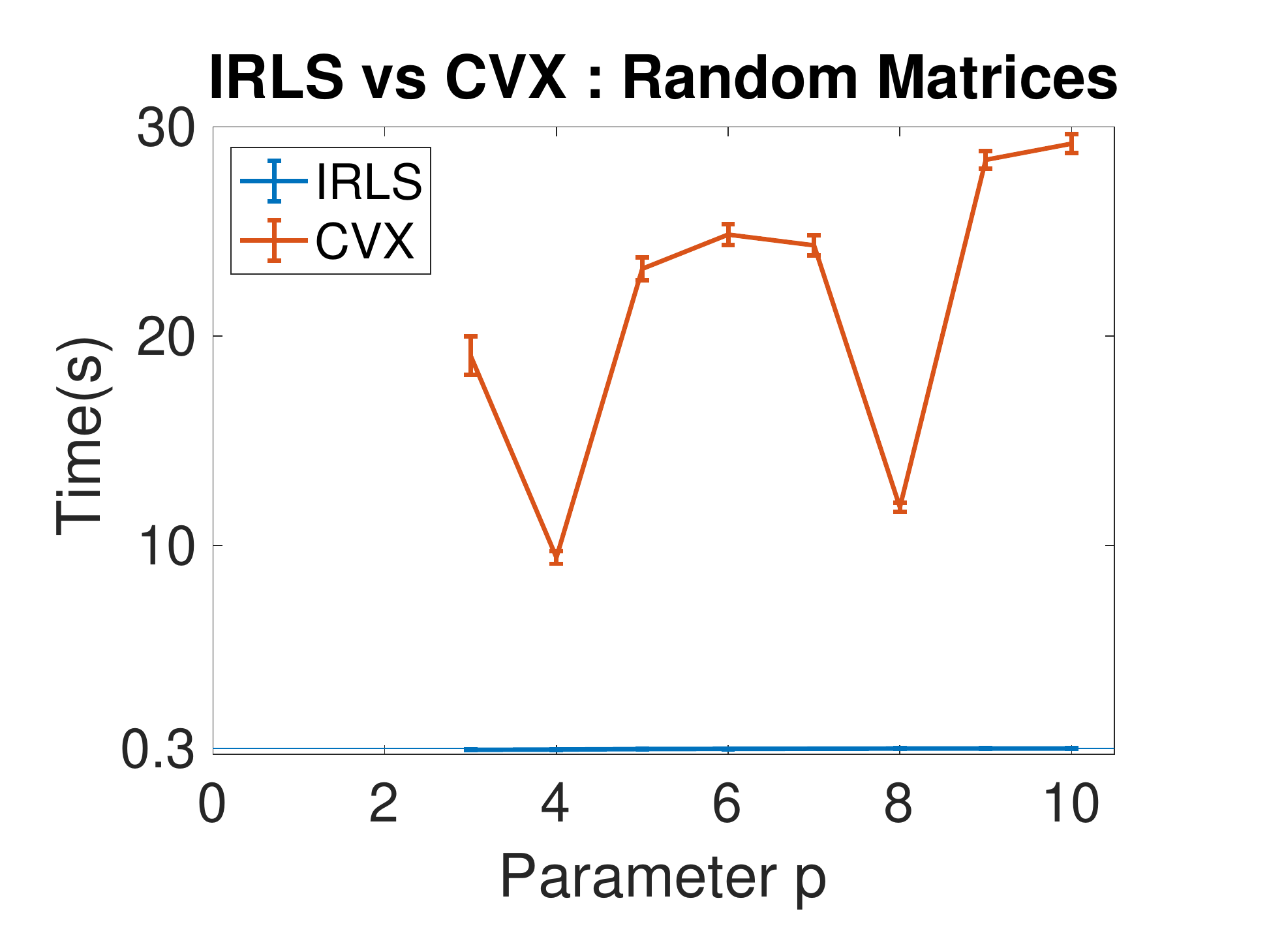}}
  \hfill
  \centering
    \subfloat[Fixed $p = 8$. The number of nodes : $50k, k = 1,2,...,10$.]{\includegraphics[width=0.23\textwidth]{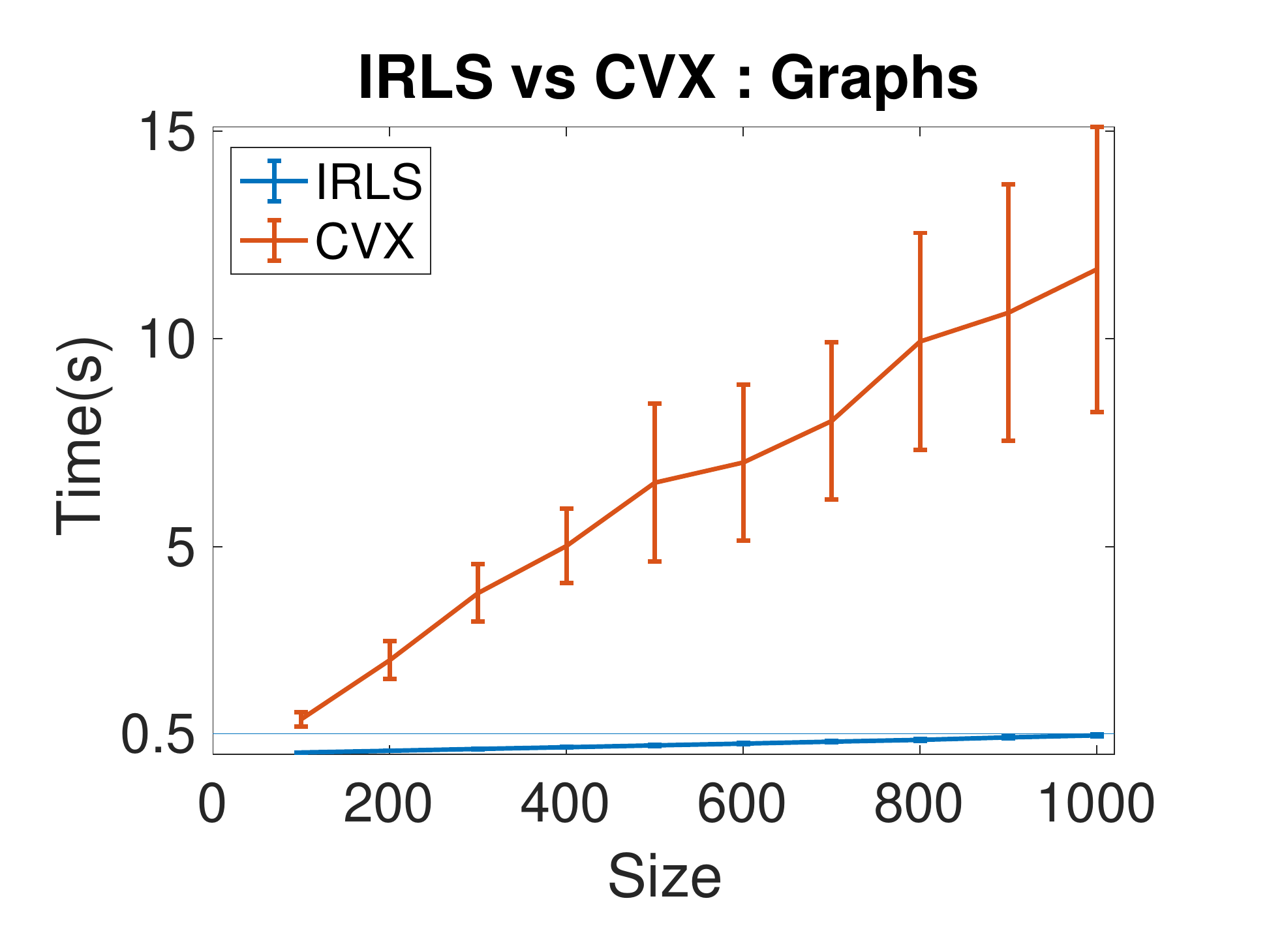}}
  \hfill
    \subfloat[Size of graphs fixed to $400$ nodes ( around $2000$ edges). ]{\includegraphics[width=0.23\textwidth]{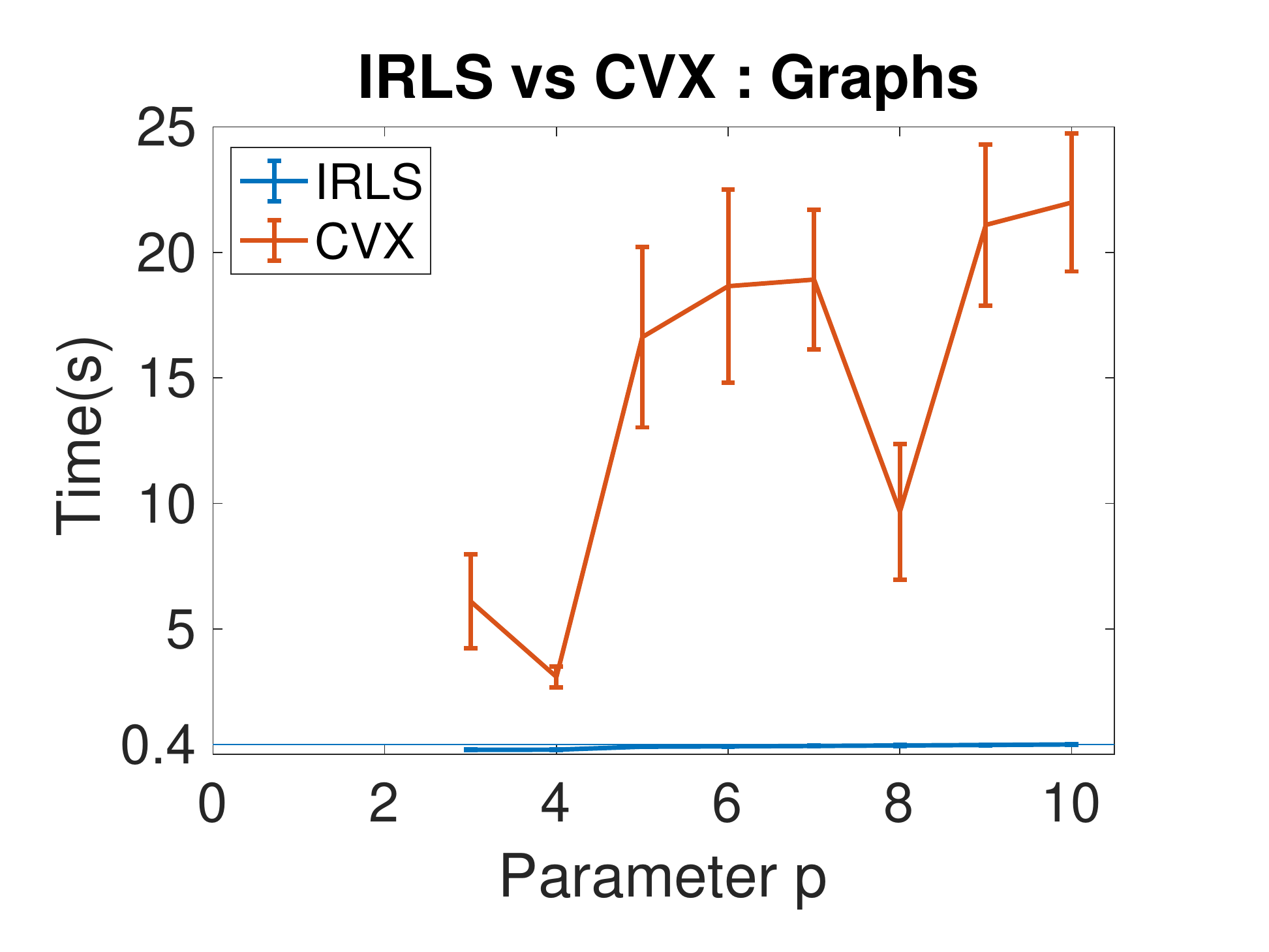}}
 \caption{Averaged over 100 samples. Precision set to $\epsilon = 10^{-8}$.CVX solver used : SDPT3 for Matrices and Sedumi for Graphs.}
 \label{fig:CVXvsIRLS}
\end{figure}

\paragraph{Instances and Parameters.} We consider two types of instances, random matrices and graphs. 
\begin{tight_enumerate}
\item {\bf Random Matrices:} We want to solve the problem $\min_{x} \norm{\AA\xx-\bb}_p$. In these instances we use random matrices $\AA$ and $\bb$, where every entry of the matrix is chosen uniformly at random between $0$ and $1$.
\item {\bf Graphs:} We use the graphs described in \cite{CalderFL19}.  The set of vertices is generated by choosing vectors in $[0,1]^{10}$ uniformly at random and the edges are created by connecting the $10$ nearest neighbours. Weights of each edge is specified by a gaussian type function (Eq 3.1,\cite{CalderFL19}). Very few vertices (around 10) have labels  which are again chosen uniformly at random between $0$ and $1$. The problem studied on these instances is to determine the minimizer of the $\ell_p$ laplacian. We formulate this problem into the form $\min_{x} \norm{\AA\xx-\bb}_p^p$, details of this formulation can be found in the Appendix that is in the supplementary material. 

\end{tight_enumerate}
Note that we have $3$ different parameters for each problem, the size of the instance i.e., the number of rows of matrix $\AA$, the norm we solve for, $p$, and the accuracy to which we want to solve each problem, $\epsilon$. We will consider each of these parameters independently and see how our algorithm scales with them for both instances.

\paragraph{Benchmark Comparisons.} We compare the performance of our program with the following:
\begin{tight_enumerate}
\item Standard MATLAB optimization package, CVX \cite{cvx, gb08}.
\item The most efficient algorithm for $\ell_p$-semi supervised learning given in \cite{CalderFL19} was newton's method with homotopy. We take their hardest problem, and compare the performance of their code with ours by running our algorithm for the same number of iterations as them and showing that we get closer to the optimum, or in other words a smaller error $\epsilon$, thus showing we converge much faster.
\end{tight_enumerate}

\paragraph{Implementation Details.} We normalize the instances by running our algorithm once and dividing the vector $\bb$ by the norm of the final objective, so that our norms at the end are around $1$. We do this for every instance before we measure the runtime or the iteration count for uniformity and to avoid numerical precision issues. All experiments were performed on MATLAB 2018b on a Desktop ubuntu machine with an Intel Core $i5$-$4570$ CPU @ $3.20GHz \times 4$ processor and 4GB RAM. For the graph instances, we fix the dimension of the space from which we choose vertices to $10$ and the number of labelled vertices to be $10$. The graph instances are generated using the code \cite{Flores19} by \cite{CalderFL19}. Other details specific to the experiment are given in the captions.

\subsection{Experimental Results}
\paragraph{Dependence on Parameters.} Figure \ref{fig:Matrices} shows the dependence of the number of iterations and runtime on our parameters for random matrices. Similarly for graph instances, Figure \ref{fig:Graphs} shows the dependence of iteration count and runtime with the parameters. As expected from the theoretical guarantees, the number of iterations and runtimes increase linearly with $\log\left(\frac{1}{\epsilon}\right)$. The dependence on size and $p$ are clearly much better in practice (nearly constant and at most linear respectively) than the theoretical bounds ($m^{1/2}$ and $p^{3.5}$ respectively) for both kinds of instances.

\paragraph{Comparisons with Benchmarks.} 
\begin{tight_itemize}
\item Figure \ref{fig:CVXvsIRLS} shows the runtime comparison between our
  IRLS algorithm {\Alg} and CVX. For all instances, we ensured that our
  final objective was smaller than the objective of the CVX solver. As
  it is clear for both kinds of instances, our algorithm takes a lot
  lesser time and also increases more slowly with size and $p$ as
  compared to CVX. Note that that CVX does a lot better when
  $p = 2^k$, but it is still at least $30$-$50$ times slower for
  random matrices and $10$-$30$ times slower for graphs.
\item Figure \ref{fig:IRLSvsNt} shows the performance of our algorithm
  when compared to the IRLS/Homotopy method of \cite{CalderFL19}. We
  use the same linear solvers for both programs, preconditioned
  conjugate gradient with an incomplete cholesky preconditioner and
  run both programs to the same number of iterations. The plots
  indicate the value $\epsilon$ as described previously. For our IRLS
  algorithm we indicate our upper bound on $\epsilon$ and for their
  procedure we indicate a lower bound on $\epsilon$ which is the
  relative difference in the objectives achieved by the two
  algorithms. It is clear that our algorithm achieves an error that is
  orders of magnitudes smaller than the error achieved by their
  algorithm. This shows that our algorithm has a much faster rate of
  convergence. Note that there is no guarantee on the convergence of
  the method used by \cite{CalderFL19}, whereas we prove that our
  algorithm converges in a small number of iterations.
\end{tight_itemize}


\section{Discussion}
To conclude, we present {\Alg}, the first IRLS algorithm that provably
converges to a high accuracy solution in a small number of
iterations. This settles a problem that has been open for over three
decades. Our algorithm is very easy to implement and we demonstrate
that it works very well in practice, beating the standard optimization
packages by large margins. The theoretical bound on the numbers of
iterations has a sub-linear dependence on size and a small polynomial
dependence on $p$, however in practice, we see an almost constant
dependence on size and at most linear dependence on $p$ in random
instances and graphs. In order to achieve the best theoretical bounds
we would require some form of acceleration. For $\ell_1$ and
$\ell_{\infty}$ regression, it has been shown that it is possible to
achieve acceleration, however without geometric convergence. It
remains an open problem to give a practical IRLS algorithm which
simultaneously has the best possible theoretical convergence bounds.


\section*{Acknowledgements}
DA is supported by SS's NSERC Discovery grant and an Ontario Graduate
Scholarship.  SS is supported by the Natural Sciences and Engineering
Research Council of Canada (NSERC), a Connaught New Researcher award,
and a Google Faculty Research award. RP is partially supported by the
NSF under Grants No. 1637566 and No. 1718533.

\bibliographystyle{alpha}
\bibliography{main.bib}
\newpage

\begin{appendices}

\section{Coordinate Wise Convergence vs Convergence in Objective Value}

For an
algorithm with a $\log \frac{1}{\eps}$ dependence of the running time
for computing a $(1+\eps)$-approximate solution, like {\Alg}, the
guarantee can be translated into a guarantee for convergence in the
solution without any significant loss in the runtime complexity of the
method. We demonstrate this theoretically and experimentally below. 

\begin{lemma}
If $\xx$ is a $(1+\delta)$-approximate solution and $\xx^{\star}$ is the optimum, then
\[\norm{\xx-\xx^{\star}}_{\infty} \le \frac{2m^{\frac{1}{2}}}{\sigma_{\min}(\AA)}\left( \frac{2\delta}{m}
  \right)^{\frac{1}{p}} \norm{\AA\xx^{\star} - \bb}_{p},\]
  where $\sigma_{\min}(\AA)$ is the smallest singular value of $\AA$.
\end{lemma}

\begin{proof}
Given that $\xx$ is a $(1+\delta)$-approximate solution, using
Lemma \ref{lem:precondition}, we can write the following lower bound on the objective
value:
\begin{align*}
  (1+\delta)\norm{\AA\xx^{\star} - \bb}_{p}^{p}
    & \ge \norm{\AA\xx^{\star} - \bb}_{p}^{p} + p\left( \AA \xx^{\star} -
      \bb\right)^{\top} \RR \AA(\xx -
      \xx^{\star})
      \\ & + \nicefrac{p}{8} \AA(\xx -
      \xx^{\star})^{\top}\AA^{\top} \RR \AA (\xx - \xx)^{\star} +
      {2^{-(p+1)}} \norm{\AA\xx-\AA\xx^{\star}}_{p}^{p},
\end{align*}
where $\RR = \text{diag}(|\AA\xx^{\star}-\bb|^{p-2}).$ Since the
gradient at $\xx^{\star}$ is 0, simplifying, we get,
$ 2^{p+1}\delta\norm{\AA\xx^{\star} - \bb}_{p}^{p} \ge
\norm{\AA\xx-\AA\xx^{\star}}_{p}^{p}.$ Now, translating between
various norms, we obtain,
\vspace{-2pt}
\begin{align*}
  \textstyle
  \norm{\xx-\xx^{\star}}_{\infty} \le  \frac{1}{\sigma_{\min}(\AA)}
  \norm{\AA\xx-\AA\xx^{\star}}_{2}
  \le  \frac{m^{\frac{1}{2} - \frac{1}{p}}}{\sigma_{\min}(\AA)}
  \norm{\AA\xx-\AA\xx^{\star}}_{p}
  \le
  \frac{2m^{\frac{1}{2}}}{\sigma_{\min}(\AA)}\left( \frac{2\delta}{m}
  \right)^{\frac{1}{p}} \norm{\AA\xx^{\star} - \bb}_p. 
\end{align*}
\end{proof}

We can achieve the guarantee
$\norm{\xx-\xx^{\star}}_{\infty} \le \eps\norm{\AA\xx^{\star} -
  \bb}_{p}$ by picking
$\delta = \left( \frac{\eps\sigma_{\min}(\AA)}{4m} \right)^{p}$. This
gives $\log \frac{m}{\delta} = O(p \log \frac{m}{\sigma_{\min}(\AA) \eps}),$
and hence a total iteration count of $O(p^{4.5} m^{\frac{p-2}{2(p-1)}} \log \frac{m}{{\sigma_{\min}(\AA)
    \eps}}).$ Asymptotically, the running time bound is only off by a
factor of $p$ if we wish to measure the convergence in
$\ell_{\infty}$-norm, as long as
$\log \frac{1}{\sigma_{\min}(\AA)} = O(\log \frac{m}{\eps}).$

\begin{figure}[H]
\includegraphics[width=0.48\textwidth]{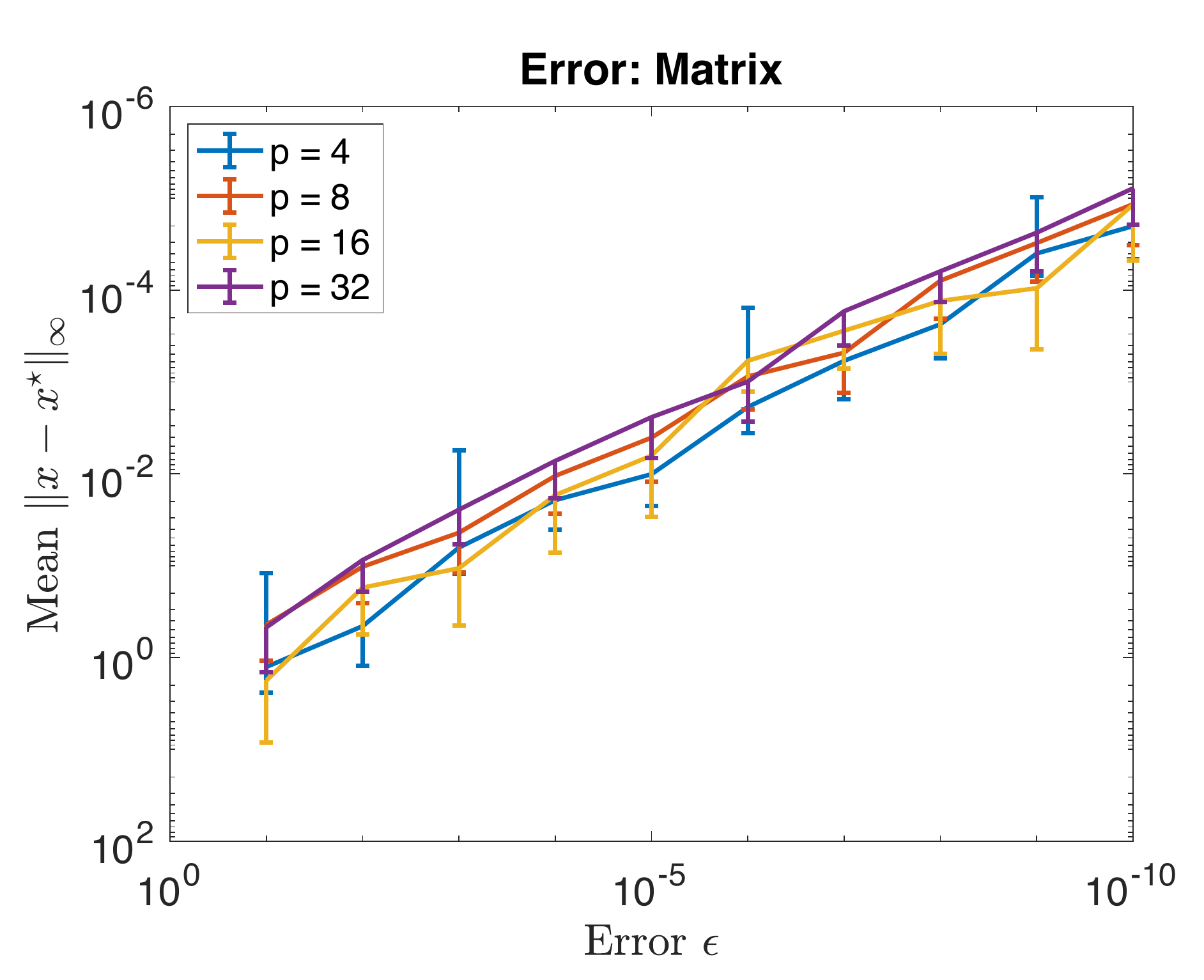}
\hfill
\includegraphics[width=0.48\textwidth]{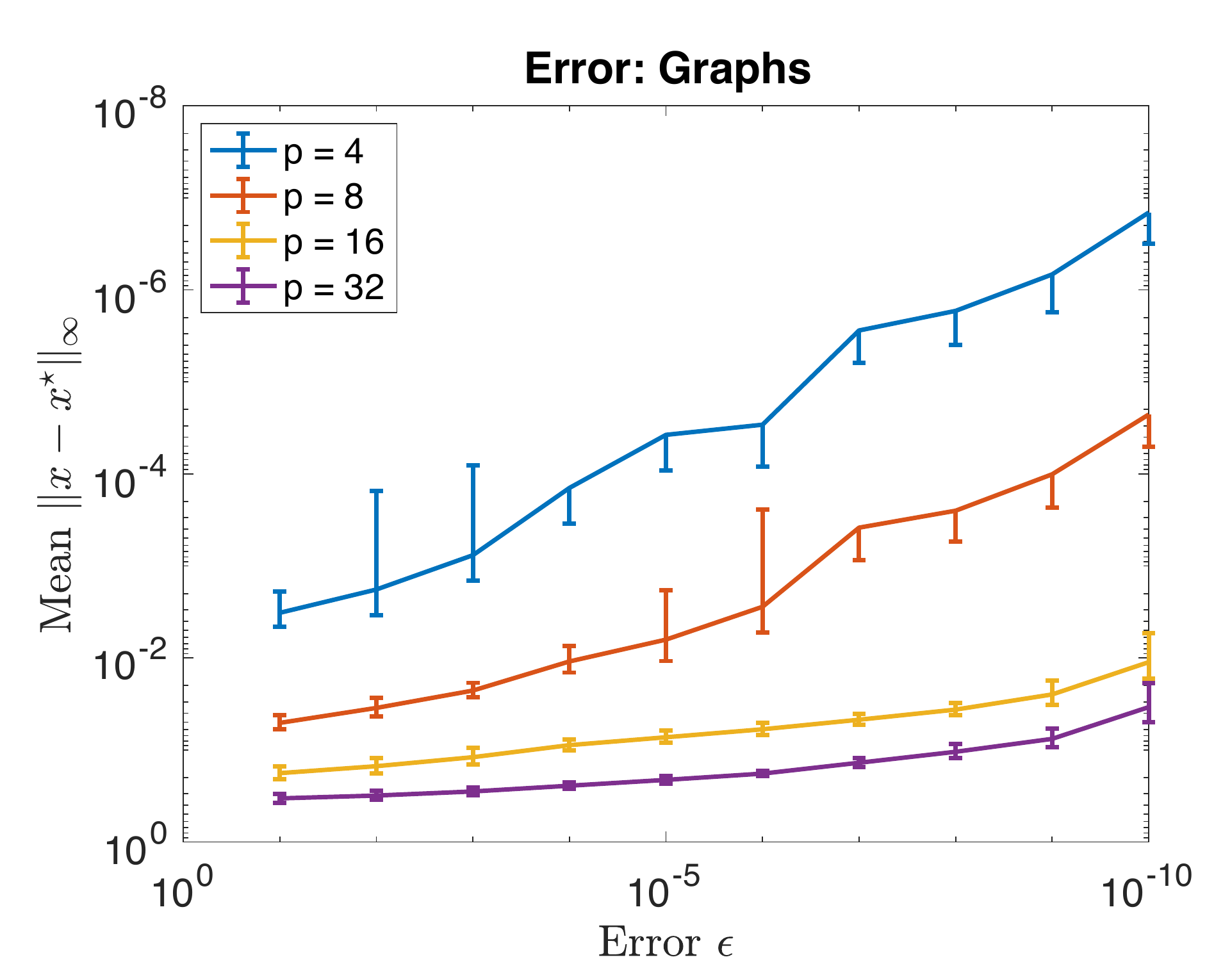}
\caption{Maximum coordinate wise difference with the optimum vs accuracy to which the objective is close to the optimum, for both graphs and random matrix instances.}
\label{fig:coordinateError}
 \end{figure}
%

 We also demonstrate this relation experimentally. The plots in Figure \ref{fig:coordinateError}
 demonstrate the average resulting $\ell_{\infty}$ norm deviation for
 the solution computed, as we change the $\eps$ parameter used in the
 algorithm. We use the instances described in the paper; matrices
 of size $1000\times 800$ and graphs with $1000$ nodes.
 For each instance, we: 1) find a very high accuracy solution, by
 choosing a very small $\eps \sim 10^{-25},$ 2) scale the problem so
 that the optimum value is $1$, and run the algorithm again to find
 the optimum solution $\xx^{\star}$. 3) Now we have a problem such
 that $\norm{\AA\xx^{\star}-\bb}_{p}=1,$ we run the algorithm again
 with various values of $\eps,$ to obtain solutions $\xx(\eps)$ and
 plot $\norm{\xx(\eps)-\xx^{\star}}_{\infty}$ (averaged over $20$
 samples). These results are very much in agreement with the
 theoretical $\eps^{\frac{1}{p}}$ dependence proved above. (Note that
 the error bars indicate $\log (\text{mean} \pm \text{std})$ so they
 are missing on one side when $\text{mean} < \text{std}.$)

\section{Proofs from Section \ref{sec:Algo}}

\subsection{Proof of Lemma \ref{lem:IterativeRefinement}}
\IterativeRefinement*
We first show that we can upper and lower bound the change in objective by a linear term plus a quadratically smoothed function.
\begin{lemma}\label{lem:precondition}
For any $\xx,\Delta$ and $p \geq 2$, we have for $\rr =|\xx|^{p-2}$ and $\gg = p |\xx|^{p-2}\xx$,
\[
\frac{p}{8} \sum_e \rr_e \Delta_e^2 + \frac{1}{2^{p+1}}\norm{\Delta}_p^p \leq \norm{\xx+\Delta}^p_p - \norm{\xx}_p^p - \gg^{\top}\Delta \leq 2 p^2 \sum_e \rr_e \Delta_e^2 + p^p \norm{\Delta}_p^p.
\]
\end{lemma}
The proof of the above lemma is long and hence deferred to the end of this section.
Applying the above lemma on our objective we get,
\begin{equation}\label{eq:precondition}
\frac{p}{8} (\AA\Delta)^{\top}\RR\AA\Delta + \frac{1}{2^{p+1}}\norm{\AA\Delta}_p^p \leq \norm{\AA(\xx+\Delta)-\bb}^p_p - \norm{\AA\xx-\bb}_p^p - \gg^{\top}\AA\Delta \leq 2 p^2 (\AA\Delta)^{\top}\RR\AA\Delta + p^p \norm{\AA\Delta}_p^p,
\end{equation}
where $\RR$ is the diagonal matrix with entries $|\AA\xx-\bb|^{p-2}$ and $\gg = p\RR(\AA\xx-\bb)$. We next show the relation between the residual problem defined in the preliminaries and the change in objective value when $\xx$ is updated by $\Delta$.
\begin{lemma}
\label{lem:RelateResidualOpt}
For any $\xx,\Delta$ and $p \geq 2$ and $\lambda = 16p$,
\[
\residual(\Delta) \leq  \norm{\AA\xx - \bb}_p^p  -\norm{\AA(\xx-\Delta) - \bb}_p^p,
\]
and 
\[
\norm{\AA\xx - \bb}_p^p  -\norm{\AA(\xx-\lambda\Delta) - \bb}_p^p \leq \lambda \residual(\Delta).
\]
\end{lemma}
\begin{proof}
The first inequality directly follows from \eqref{eq:precondition}. For the second inequality,
\begin{align*}
\norm{\AA\xx - \bb}_p^p  -\norm{\AA(\xx-\lambda\Delta) - \bb}_p^p & \leq  \lambda \gg^{\top}\Delta - \lambda^2\frac{p}{8} \Delta^{\top}\AA^{\top}\RR\AA\Delta  - \lambda^p\frac{1}{2^{p+1}}\norm{\AA\Delta}_p^p\\
& = \lambda \left(\gg^{\top}\Delta - \lambda \frac{p}{8} \Delta^{\top}\AA^{\top}\RR\AA\Delta  - \lambda^{p-1}\frac{1}{2^{p+1}} \norm{\AA\Delta}_p^p\right)\\
& \leq \lambda \left(\gg^{\top}\AA\Delta -2p^2 \Delta^{\top}\AA^{\top}\RR\AA\Delta  - p^p\norm{\AA\Delta}_p^p \right).
\end{align*}
\end{proof}

\subsubsection{Proof of Lemma Iterative Refinement}
\begin{proof}
Let $\Dtil$ be a $\kappa$-approximate solution to the residual problem. Using this fact and Lemma \ref{lem:RelateResidualOpt} for $\Delta = \frac{\xx - \xx^{\star}}{\lambda}$, we get,
\[
\residual(\Dtil) \geq \frac{1}{\kappa}\residual(\Dopt) \geq \frac{1}{\kappa}\residual\left(\frac{\xx - \xx^{\star}}{\lambda}\right) \geq  \frac{1}{\lambda \kappa}\left(\norm{\AA\xx-\bb}_p^p - OPT\right).
\]
Also,
\begin{align*}
\norm{\AA(\xx-\Dtil)-\bb}_p^p - OPT &\leq \norm{\AA\xx-\bb}_p^p - \residual(\Dtil) - OPT\\
& \leq  \left(\norm{\AA\xx-\bb}_p^p - OPT\right) -  \frac{1}{\lambda \kappa}\left(\norm{\AA\xx-\bb}_p^p -OPT\right)\\
& = \left(1 -  \frac{1}{\lambda \kappa}\right) \left(\norm{\AA\xx-\bb}_p^p - OPT\right).
\end{align*}
Now, after $t$ iterations,
\[
\norm{\AA(\xx^{(t)}-\Dtil)-\bb}_p^p - OPT \leq \left(1 -  \frac{1}{\lambda \kappa}\right)^t \left(\norm{\AA\xx^{(0)}-\bb}_p^p - OPT\right) \leq \left(1 -  \frac{1}{\lambda \kappa}\right)^t m^{(p-2)/2} OPT.
\]
Thus, for our value of $\lambda=16p,$  $8p^2 \kappa \log (m/\epsilon)$
iterations suffice to obtain a
$(1+\eps)$-approximate solution.
\end{proof}

\subsubsection{Proof of Lemma \ref{lem:precondition}}
\begin{proof}
To show this, we show that the above holds for all coordinates. For a single coordinate, the above expression is equivalent to proving,
\[
\frac{p}{8} |x|^{p-2} \Delta^2 + \frac{1}{2^{p+1}}\abs{\Delta}^p \leq \abs{\xx+\Delta}^p - \abs{\xx}^p - p\abs{x}^{p-1}sgn(x)\Delta \leq 2 p^2 |x|^{p-2} \Delta^2  + p^p \abs{\Delta}^p.
\]
Let $\Delta = \alpha x$. Since the above clearly holds for $x=0$, it remains to show for all $\alpha$,
\[
\frac{p}{8}\alpha^2 + \frac{1}{2^{p+1}}\abs{\alpha}^p \leq \abs{1+\alpha}^p - 1 - p\alpha \leq 2 p^2 \alpha^2  + p^p \abs{\alpha}^p.
\]
\begin{enumerate}
\item $\alpha \geq 1$:\\
\noindent In this case, $1+\alpha \leq 2 \alpha \leq p\cdot \alpha$. So, $ \abs{1+\alpha}^p \leq p^p \abs{\alpha}^p$ and the right inequality directly holds. To show the other side, let 
\[
h(\alpha) = (1+\alpha)^p - 1 - p\alpha - \frac{p}{8} \alpha^2 - \frac{1}{2^{p+1}}{\alpha}^p.
\]
We have,
\[
h'(\alpha) = p(1+\alpha)^{p-1}  - p- \frac{p}{4} \alpha - \frac{p}{2^{p+1}}{\alpha}^{p-1}
\]
and 
\[
h''(\alpha) = p(p-1)(1+\alpha)^{p-2}  - \frac{p}{4}  - \frac{p(p-1)}{2^{p+1}}{\alpha}^{p-2} \geq 0.
\]
Since $h''(\alpha) \geq 0$, $h'(\alpha) \geq h'(1) \geq 0$. So $h$ is an increasing function in $\alpha$ and $h(\alpha) \geq h(1) \geq 0$.

\item $\alpha \leq -1$:\\
Now, $\abs{1+\alpha} \leq 1+\abs{\alpha} \leq p\cdot \abs{\alpha}$, and $2\alpha^2 p^2 - \abs{\alpha} p \geq 0$. As a result,
\[
\abs{1+\alpha}^p \leq - \abs{\alpha} p +2\alpha^2 p^2  + p^p\cdot \abs{\alpha}^p
\]
which gives the right inequality. Consider,
\[
h(\alpha) = |1+\alpha|^p - 1 - p\alpha - \frac{p}{8} \alpha^2 - \frac{1}{2^{p+1}}|\alpha|^p.
\]
\[
h'(\alpha) = - p|1+\alpha|^{p-1}  - p - \frac{p}{4} \alpha + p\frac{1}{2^{p+1}}|\alpha|^{p-1}.
\]
Let $\beta = -\alpha$. The above expression now becomes,
\[
- p (\beta - 1)^{p-1} - p + \frac{p}{4} \beta + p\frac{1}{2^{p+1}}\beta^{p-1}.
\]
We know that $\beta \geq 1$. When $\beta \geq 2$, $\frac{\beta}{2} \leq \beta - 1$ and $\frac{\beta}{2} \leq \left(\frac{\beta}{2}\right)^{p-1}$. This gives us,
\[
\frac{p}{4} \beta + p\frac{1}{2^{p+1}}\beta^{p-1} \leq \frac{p}{2} \left(\frac{\beta}{2}\right)^{p-1} +  \frac{p}{2} \left(\frac{\beta}{2}\right)^{p-1} \leq p (\beta - 1)^{p-1}
\]
giving us $h'(\alpha) \leq 0$ for $\alpha \leq -2$. When $\beta \leq 2$, $\frac{\beta}{2} \geq \left(\frac{\beta}{2}\right)^{p-1}$ and $\frac{\beta}{2}  \leq 1$.
\[
\frac{p}{4} \beta + p\frac{1}{2^{p+1}}\beta^{p-1} \leq \frac{p}{2}\cdot  \frac{\beta}{2} + \frac{p}{2}\cdot  \frac{\beta}{2} \leq p
\]
giving us $h'(\alpha) \leq 0$ for $ -2 \leq \alpha \leq -1$. Therefore, $h'(\alpha) \leq 0$ giving us, $h(\alpha) \geq h(-1) \geq 0$, thus giving the left inequality.

\item $\abs{\alpha} \leq 1$:\\
Let $s(\alpha) =   1 + p\alpha + 2p^2 \alpha^2 +  p^p \abs{\alpha}^p - (1+\alpha)^p.$ Now,
\[
s'(\alpha)  = p + 4 p^2 \alpha +  p^{p+1} \abs{\alpha}^{p-1} sgn(\alpha) - p (1+\alpha)^{p-1}.
\]
When $\alpha \leq 0$, we have,
\[
s'(\alpha)  = p + 4 p^2 \alpha - p^{p+1} \abs{\alpha}^{p-1} - p (1+\alpha)^{p-1}.
\]
and 
\[
s''(\alpha)  = 4 p^2 +  p^{p+1}(p-1) \abs{\alpha}^{p-2} - p (p-1)(1+\alpha)^{p-1} \geq 2 p^2 +  p^{p+1}(p-1) \abs{\alpha}^{p-2} - p (p-1) \geq 0 .
\]
So $s'$ is an increasing function of $\alpha$ which gives us, $s'(\alpha) \leq s'(0) = 0$. Therefore $s$ is a decreasing function, and the minimum is at $0$ which is $0$. This gives us our required inequality for $\alpha \leq 0$.
When $\alpha \geq \frac{1}{p-1}$, $1 + \alpha \leq p \cdot \alpha$ and $s'(\alpha) \geq 0$. We are left with the range $0 \leq \alpha \leq \frac{1}{p-1}$. Again, we have,
\begin{align*}
s''(\alpha)  & = 4 p^2 +  p^{p+1}(p-1) \abs{\alpha}^{p-2} - p (p-1)(1+\alpha)^{p-1} \\
& \geq 4 p^2 +  p^{p+1}(p-1) \abs{\alpha}^{p-2} - p (p-1) (1+\frac{1}{p-1})^{p-1}\\
&\geq 4 p^2 +  p^{p+1}(p-1) \abs{\alpha}^{p-2} - p (p-1) e, \text{When $p$ gets large the last term approaches $e$}\\
& \geq 0.
\end{align*}
Therefore, $s'$ is an increasing function, $s'(\alpha) \geq s'(0) = 0$. This implies $s$ is an increasing function, giving, $s(\alpha) \geq s(0)=0$ as required. 

To show the other direction,
\[
h(\alpha) = (1+\alpha)^p - 1 - p\alpha - \frac{p}{8} \alpha^2 - \frac{1}{2^{p+1}}\abs{\alpha}^p \geq (1+\alpha)^p - 1 - p\alpha - \frac{p}{8} \alpha^2 - \frac{p}{8}{\alpha}^2 = (1+\alpha)^p - 1 - p\alpha - \frac{p}{4} \alpha^2.
\]
Now, since $p \geq 2$,
\begin{align*}
&\left((1+\alpha)^{p-2} - 1 \right)sgn(\alpha) \geq 0\\
\Rightarrow &\left((1+\alpha)^{p-1} - 1 - \alpha \right)sgn(\alpha) \geq 0\\
\Rightarrow &\left(p(1+\alpha)^{p-1} - p - \frac{p}{2}\alpha \right)sgn(\alpha) \geq 0
\end{align*}
We thus have, $h'(\alpha) \geq 0$ when $\alpha$ is positive and $h'(\alpha) \leq 0$ when $\alpha$ is negative. The minimum of $h$ is at $0$ which is $0$. This concludes the proof of this case.
\end{enumerate}
\end{proof}

\subsection{Proof of Lemma that Checks Progress in Objective}

We will next prove the following Lemma which shows that we do not change $i$ when we have the correct value of $i$.

\begin{restatable}{lemma}{Progress}(Check Progress).
\label{lem:CheckProgress}
Let $\alpha_0$ be as defined in line \eqref{alg:line:DefineAlpha} of Algorithm \ref{alg:CheckProgress} and $\Dtil$ the solution of program \eqref{eq:ProblemSolve}. If $ i/2 <  \frac{(\|\AA\xx^{(t)}-\bb\|_p^p -  \|\AA\xx^{\star}-b\|_p^p)}{16p} \leq i$, then $\residual(\alpha_0 \cdot \Dtil) \geq \frac{\alpha_0 i}{4}$ and $(\AA\Dtil)^{\top}(\RR+s\II)\AA\Dtil \leq \lambda i /p^2$.
\end{restatable}

We require bounding the objective of program \ref{eq:ProblemSolve}. To do that we first give a bound on a decision version of the residual problem, and then relate this problem with problem \ref{eq:ProblemSolve}.
\begin{lemma}
\label{lem:OPT2Decision}
Let $i$ be such that the optimum of the residual problem, $\residual(\Dopt) \in (i/2,\lambda i]$. Then the following problem has optimum at most $\lambda i$.
\begin{align}
\label{eq:Decision}
\begin{aligned}
\min_{\Delta \in \mathbb{R}^m} \quad& 2p^2 (\AA\Delta)^{\top}\RR\AA\Delta +p^p\norm{\AA\Delta}_p^p\\
& \gg^{\top}\AA \Delta = i/2\\
& \CC\Delta = 0.
\end{aligned}
\end{align}
\end{lemma}
\begin{proof}
  The assumption on the residual is
\[
\residual(\Dopt) = \gg^{\top}\AA \Dopt -
2p^2(\AA\Dopt)^{\top}\RR\AA\Dopt - p^p\norm{\AA\Dopt}_p^p \in  (i/2,\lambda i].
\]
Since the last $2$ terms are strictly non-positive, we must have,
$\gg^{\top}\AA \Dopt  \geq i/2.$
Since $\Dopt$ is the optimum and satisfies $\CC\Dopt = 0$, 
\[
\frac{d}{d\lambda}\left(\gg^{\top}\lambda \AA\Dopt - 2p^2 \lambda^2(\AA\Dopt)^{\top}\RR\AA\Dopt - \lambda^p p^p\norm{\AA\Dopt}_p^p \right)_{\lambda = 1} =0.
\]
Thus,
\[
\gg^{\top}\AA \Dopt -  2 p^2(\AA\Dopt)^{\top}\RR\AA\Dopt  -  p^{p}\norm{\AA\Dopt}_p^p  = 2 p^2(\AA\Dopt)^{\top}\RR\AA\Dopt  + (p-1) p^{p}\norm{\AA\Dopt}_p^p.
\]
\deeksha{Does not hold for $p<2$\\}
Since $p \ge 2,$ we get the following
\[
2 p^2(\AA\Dopt)^{\top}\RR\AA\Dopt  +  p^{p}\norm{\AA\Dopt}_p^p \leq \gg^{\top}\AA \Dopt -  2 p^2(\AA\Dopt)^{\top}\RR\AA\Dopt  -  p^{p}\norm{\AA\Dopt}_p^p \leq \lambda i.
\]
For notational convenience, let function $ h_p(\rr,\Delta) = 2 p^2(\AA\Delta)^{\top}\RR\AA\Delta  +  p^{p}\norm{\AA\Delta}_p^p$.
Now, we know that, $\gg^{\top}\AA \Dopt  \geq i/2$ and $ \gg^{\top} \AA\Dopt - h_p(\rr,\Dopt) \leq \lambda i$. This gives, 
\[
i/2 \leq \gg^{\top} \AA\Dopt \leq h_p(\rr,\Dopt) + \lambda i \leq 2\lambda i.
\]
Let $\Delta = \delta \Dopt$, where $\delta = \frac{i}{2\gg^{\top}\AA\Dopt}$. Note that $\delta \in [1/4\lambda, 1]$. Now, $\gg^{\top}\AA\Delta = i/2$ and, 
\[
h_p(\rr,\Delta) \leq \max\{\delta^2,\delta^p\} h_p(\rr,\Dopt) \leq \lambda i.
\]
Note that this $\Delta$ satisfies the constraints of program \eqref{eq:Decision} and has an optimum at most $\lambda i$. So the optimum of the program must have an objective at most $\lambda i$.
\end{proof}

 \begin{claim}\label{cl:RelateObj}
 If the optimal objective of program \eqref{eq:Decision} is at most $Z$, then the optimum objective of program \eqref{eq:ProblemSolve} is at most $\frac{Z}{2p^2}+\frac{i^{(p-2)/p}Z^{2/p}}{2p^2}$.
 \end{claim}
 \begin{proof}
 Let $\Dopt$ denote the optimizer of \eqref{eq:Decision} and $\Dtil$ be the optimizer of \eqref{eq:ProblemSolve}. Since the optimum objective of \eqref{eq:Decision} is at most $Z$, we have we have $\norm{\AA\Dopt}_p^p \leq \frac{Z}{p^p} $. This implies that $\Dopt^{\top}\AA^{\top} \AA\Dopt \leq \frac{Z^{2/p}}{p^2} m^{(p-2)/p}$.
Since $\Dopt$ is a feasible solution of \eqref{eq:ProblemSolve}, we have for our value of $s$,
\[
\Dtil^{\top}\AA^{\top}(\RR +  s^{(t)}\II) \AA\Dtil \leq \Dopt^{\top}\AA^{\top}(\RR +  s^{(t)}\II) \AA\Dopt \leq \frac{Z}{2p^2}+\frac{i^{(p-2)/p}Z^{2/p}}{2p^2}.
\]
\deeksha{For $p<2$, the constants change to $\lambda^{2/p}$. Do not
  need the poly(m) factor in $s$ anymore.}
 \end{proof}


\subsection*{Proof of Lemma \ref{lem:CheckProgress}}
\begin{proof}
Since, 
\[
 i/2 <  \frac{(\|\AA\xx^{(t)}-\bb\|_p^p -  \|\AA\xx^{\star}-b\|_p^p)}{16p} \leq i,
\] 
we know that the optimum of the residual problem lies between $(i/2,
\lambda i]$ (Lemma \ref{lem:RelateResidualOpt}). From Lemma~\ref{lem:OPT2Decision} the optimum of the problem \eqref{eq:Decision} is at most $\lambda i$. Now, from Claim \ref{cl:RelateObj}, we know that $(\AA
\Dtil)^{\top} (\RR+s\II) (\AA \Dtil) \leq \lambda i/p^2$. Also, note that $\alpha_0 + k \alpha_0^{p-1} \leq \frac{1}{8\lambda}$. Consider the following,
\begin{align*}
\residual(\alpha_0 \cdot \Dtil) = &\alpha_0 \gg^{\top}\AA\Dtil  - \alpha_0^22p^2 (\AA \Dtil)^{\top} \RR (\AA \Dtil)  - \alpha_0^p p^p\norm{\AA\Dtil}_p^p\\
& \geq \alpha_0 \gg^{\top}\AA\Dtil - (\alpha_0^2 + k \alpha_0^p) 2p^2 (\AA \Dtil)^{\top} (\RR+s\II) (\AA \Dtil)\\
& = \alpha_0 \left(\gg^{\top}\AA\Dtil - (\alpha_0 + k \alpha_0^{p-1}) 2p^2(\AA \Dtil)^{\top} (\RR+s\II) (\AA \Dtil)\right)\\
& \geq  \alpha_0\left(\gg^{\top}\AA\Dtil -  \frac{1}{8\lambda}2p^2 (\AA \Dtil)^{\top}(\RR+s\II) (\AA \Dtil)\right) \\
& \geq \alpha_0\left(\frac{i}{2}-  \frac{1}{8\lambda} 2 \lambda i \right)\\
& \geq \frac{\alpha_0}{4} i
\end{align*}
\deeksha{A version of this with slightly different $\alpha_0$ holds for $p<2$.}
\end{proof}

\subsection{Proof of Lemma \ref{lem:Invariant}}
\Invariant*
\begin{proof}
We use induction to show this. Initially we set, $i = \|\AA\xx^{(0)} -\bb\|_p^p/16p$. When the optimum is not $0$, this is greater than $(\|\AA\xx^{(0)} -\bb\|_p^p - \|\AA\xx^{\star} -\bb\|_p^p )/16p$. When the optimum is $0$, the initial solution (2-norm minimizer) will also give zero and we can stop our procedure. Therefore, the claim holds for $t = 1$. Suppose at iteration $t$ the claim holds. Since the objective is non-increasing, we know that,
\[
\|\AA\xx^{(t)} -\bb\|_p^p - \|\AA\xx^{\star} -\bb\|_p^p \geq \|\AA\xx^{(t+1)} -\bb\|_p^p - \|\AA\xx^{\star} -\bb\|_p^p.
\] 
Let $\Dtil$ denote the solution returned in iteration $t+1$. At iteration $t+1$, if $ (\|\AA\xx^{(t+1)} -\bb\|_p^p - \|\AA\xx^{\star} -\bb\|_p^p )/16p \in (i/2,i]$, from Lemma \ref{lem:CheckProgress}, we will always have $\residual(\alpha_0 \Dtil) \geq \alpha_0 i/4$ and $(\AA\Dtil)^{\top}(\RR+s\II)\AA\Dtil \leq \lambda i/p^2$.
So the algorithm does not reduce $i$ and as a result our claim holds for $t+1$. Otherwise, we know that $ (\|\AA\xx^{(t+1)} -\bb\|_p^p - \|\AA\xx^{\star} -\bb\|_p^p )/16p \leq i/2$ and the algorithm might reduce $i$ by half if either of the two conditions are true. However, the claim still holds.  Therefore, $i$ is always at least $ (\|\AA\xx^{(t+1)} -\bb\|_p^p - \|\AA\xx^{\star} -\bb\|_p^p )/16p $. 

We start with a solution $\xx^{(0)}$ that minimizes the $\ell_2$ norm. Therefore, the following holds,
\[
\|\AA\xx^{(0)}-\bb\|_p^p \leq \|\AA\xx^{\star}-\bb\|_p^p m^{(p-2)/2}.
\]
The value of $i$ is the minimum at termination. Therefore, it is sufficient to prove the above bound for the termination condition. Our condition gives us the following at termination (the left inequality holds because otherwise, we would have terminated in the previous iteration).
\[
i \geq \frac{\epsilon}{16p(1+\epsilon)} \|\AA\xx-\bb\|_p^p \geq  \frac{i}{2}.
\]
This implies,
\[
i \geq \frac{\epsilon}{16p(1+\epsilon)} \|\AA\xx^{\star}-\bb\|_p^p \geq \frac{\epsilon}{16p(1+\epsilon)} \|\AA\xx^{(0)}-\bb\|_p^p m^{-(p-2)/2}.
\]
We next prove the second claim. Initially our solution satisfies $\CC\xx^{(0)} = \dd$. Assuming the condition holds at iteration $t$, we show it for $t+1$. At every iteration we solve for $\Dtil$ under the constraint $\CC\Dtil= 0$. Our update rule, $\xx^{(t+1)} = \xx^{(t)} - \alpha \Dtil$ gives us,
\[
\CC\xx^{(t+1)} = \CC\xx^{(t)} - \alpha \CC\Dtil = \dd - \alpha \cdot 0 = \dd.
\]
\deeksha{the bound on $i_{min}$ has a $|p-2|$ in the power of $m$ for $p<2$. Everything else holds assuming other lemmas used in this proof hold.}
\end{proof}

\subsection{Proof of Lemma \ref{lem:Approximation}}
Our approximation depends on the quantity $\alpha_0$ which is defined in the algorithm. This depends on the value of $k$, the ratio of the $p$-norm term to the square term. Therefore, in order to bound the approximation, we first give a bound on $k$.
\begin{lemma}\label{lem:Boundk}
Let $\Dtil$ the optimum of \eqref{eq:ProblemSolve} and let $ k = \frac{p^p\norm{\AA\Dtil}_p^p}{2p^2\Dtil^{\top}\AA^{\top}(\RR+s\II) \AA \Dtil}$. If the optimum of \eqref{eq:ProblemSolve} is at most $\lambda i/p^2$, then $k$ is at most $(32pm)^{(p-2)/2}$ for $\lambda = 16p$. Let $\alpha_0 = \min\left\{\frac{1}{16\lambda}, \frac{1}{(16\lambda k)^{1/(p-1)}}\right\}$. Then $\alpha_0 \geq \Omega(p^{-1/2}m^{-\frac{(p-2)}{2(p-1)}})$ when $p \leq m$.
\end{lemma}
\begin{proof}
Since, $ s \II \preceq \RR+s\II $,
\[
\norm{\AA\Dtil}_2^2 =\Dtil^{\top} \AA^{\top} \AA \Dtil \leq \frac{1}{ s} \Dtil^{\top} \AA^{\top}(\RR+s\II) \AA \Dtil
\]
and,
\[
\norm{\AA\Dtil}_p^p \leq \norm{\AA\Dtil}_2^p \leq  \frac{1}{ s} \left(\Dtil^{\top} \AA^{\top} \AA \Dtil\right)^{(p-2)/2} \Dtil^{\top} \AA^{\top}(\RR+s\II) \AA \Dtil.
\]
\deeksha{Above not true for $p<2$, can probably be fixed by removing the poly(m) factor from $s$ and having that here instead.\\}
We also have, $\Dtil^{\top}\AA^{\top} \AA\Dtil \leq \frac{2\lambda}{p^2} m^{(p-2)/p}i^{2/p}$. Combining these,
\begin{align*}
\frac{p^p\norm{\AA\Dtil}_p^p}{2p^2\Dtil^{\top}\AA^{\top}(\RR +s\II)\AA \Dtil} & \leq \frac{p^p}{2p^2} \frac{ 2 m^{(p-2)/p}}{i^{(p-2)/p}} \left(\frac{2\lambda}{p^2} m^{(p-2)/p}i^{2/p}\right)^{(p-2)/2}\\
& \leq (4\sqrt{2})^{p-2}p^{(p-2)/2}m^{(p-2)/p} m^{(p-2)^2/2p}\\
& =(32pm)^{(p-2)/2} 
\end{align*}
We can now find a bound on $\alpha_0$.
\begin{align*}
\alpha_0 &\geq \min\left\{\Omega\left(\frac{1}{p}\right), \Omega\left(\frac{1}{p^{1/(p-1)}(pm)^{(p-2)/2(p-1)}}\right)\right\}\\
& \geq \min\left\{\Omega\left(\frac{1}{p}\right), \Omega\left(\frac{1}{p^{1/2}m^{(p-2)/2(p-1)}}\right)\right\}\\
& \geq \Omega\left(\frac{1}{p^{1/2}m^{(p-2)/2(p-1)}}\right) \text{, assuming $p \leq m$.}
\end{align*}
\deeksha{We might have to add this in the assumption , that p is less than m.}
\end{proof}

\Approximation*

\begin{proof}
In the algorithm we choose $\alpha$ such that given $\Dtil$, $\alpha = argmin_\delta \|\AA(\xx - \delta \Dtil) -\bb\|_p^p$. From our assumption, we also know that $\residual(\alpha_0\Dtil) \geq \frac{\alpha_0}{4} i$. Now, since the residual function is a convex function with value zero at the zero vector, we know that $\residual(\Dtil/\lambda) \leq \frac{1}{\lambda}\residual(\Dtil)$ (for our value of $\lambda = 16p$).
\begin{align*}
\residual(\alpha \Dtil) & \geq \lambda \residual(\alpha \Dtil/\lambda) \\
&\geq \|\AA\xx -\bb\|_p^p - \|\AA(\xx - \alpha \Dtil) -\bb\|_p^p  \\
& \geq \|\AA\xx -\bb\|_p^p - \|\AA(\xx - \alpha_0 \Dtil) -\bb\|_p^p \\
& \geq \residual(\alpha_0 \cdot \Dtil) \\
&\geq \frac{\alpha_0}{4} i  \\
&\geq \frac{\alpha_0}{4\lambda} OPT.
\end{align*}
 The last inequality follows form the fact that the optimum of the residual problem is at most $\lambda i$. This is because, Lemma \ref{lem:Invariant} shows that, $(\|\AA\xx^{(t)}-\bb\|_p^p- \|\AA\xx^{\star}-\bb\|_p^p)/16p<i$. Now from Lemma \ref{lem:RelateResidualOpt} we can conclude that the residual problem has optimum at most $\lambda i$. Since we have from our assumption that the objective of \eqref{eq:ProblemSolve} is at most $\lambda i/p^2$, from Lemma~\ref{lem:Boundk}, we can bound the factor
$O(\lambda/\alpha_0) \leq O(p^{1 +
  \frac{p-2}{2(p-1)}}m^{\frac{p-2}{2(p-1)}}) \leq
O(p^{1.5}m^{\frac{p-2}{2(p-1)}})$.

\end{proof}

\subsection{Proof of Lemma \ref{lem:Termination}}
\termination*

\begin{proof}
  We first show the forward implication.
From the assumptions, we have,
\begin{align*}
&\frac{ \norm{\AA\xx^{(t)}-\bb}_p^p -OPT}{16p} \leq i \leq \frac{\epsilon}{16p(1+\epsilon)} \norm{\AA\xx^{(t)}-\bb}_p^p\\
\Rightarrow & \norm{\AA\xx^{(t)}-\bb}_p^p \frac{1}{1+\epsilon} \leq OPT\\
\Rightarrow & \norm{\AA\xx^{(t)}-\bb}_p^p \leq (1+ \epsilon) OPT.
\end{align*}
For the other direction we have,
\[
\frac{\norm{\AA\xx^{(t)}-\bb}_p^p}{1+\epsilon} \leq OPT.
\]
Thus,
\begin{align*}
i &\leq 2\frac{ \norm{\AA\xx^{(t)}-\bb}_p^p -OPT}{16p}  \leq 2\frac{ \norm{\AA\xx^{(t)}-\bb}_p^p }{16p}\left(1 - \frac{1}{1+\epsilon}\right) \leq \frac{2\epsilon}{1+\epsilon}\frac{ \norm{\AA\xx^{(t)}-\bb}_p^p }{16p}.
\end{align*}\end{proof}

\section{Converting $\ell_p$-Laplacian Minimization to Regression Form}
Define the following terms:
\begin{tight_itemize}
\item $n$ denote the number of vertices.
\item $l$ denote the number of labels.
\item $\BB$ denote the edge-vertex adjacency matrix.
\item $\gg$ denote the vector of labels for the $l$ labelled vertices.
\item $\WW$ denote the diagonal matrix with weights of the edges.
\end{tight_itemize}
Set $\AA = \WW^{1/p}\BB$ and $\bb = - \BB[:,n:n+l]\gg$. Now $\norm{\AA\xx-\bb}_p^p$ is equal to the $\ell_p$ laplacian and we can use our IRLS algorithm to find the $\xx$ that minimizes this.


\section{Solving $\ell_2$ Problems under Subspace Constraints}
\subsection{Finding the Initial Solution}
We want to solve:
\begin{align*}
\min_{\xx} &\quad \norm{\AA\xx-\bb}_2^2\\
&\CC\xx = \dd.
\end{align*}
Using Lagrangian duality and noting that strong duality holds, we can write the above as,
\begin{align*}
L(\xx,\vv) = & \min_{\xx} \max_{\vv} \quad(\AA\xx-\bb)^{\top}(\AA\xx-\bb) + \vv^{\top}(\dd-\CC\xx)\\
= & \max_{\vv}\min_{\xx} \quad(\AA\xx-\bb)^{\top}(\AA\xx-\bb) + \vv^{\top}(\dd-\CC\xx).
\end{align*}
We first find $\xx^{\star}$ that minimizes the above objective by setting the gradient with respect to $\xx$ to $0$. We thus have,
\[
\xx^{\star} = (\AA^{\top}\AA)^{-1}\left(\frac{2\AA^{\top}\bb+\CC^{\top}\vv}{2}\right).
\]
Using this value of $\xx$ we arrive at the following dual program.
\[
L(\vv) = \max_{\vv} \quad -\frac{1}{4}\vv^{\top}\CC(\AA^{\top}\AA)^{-1}\CC^{\top}\vv -\bb^{\top}\AA(\AA^{\top}\AA)^{-1}\AA^{\top}\bb - \vv^{\top}\CC(\AA^{\top}\AA)^{-1}\AA^{\top}\bb +\bb^{\top}\bb+\vv^{\top}\dd,
\]
which is optimized at,
\[
\vv^{\star}= 2\left(\CC(\AA^{\top}\AA)^{-1}\CC^{\top}\right)^{-1} \left(\dd - \CC(\AA^{\top}\AA)^{-1}\AA^{\top}\bb\right).
\] 
Strong duality also implies that $L(\xx,\vv^{\star})$ is optimized at $\xx^{\star}$, which gives us,
\[
\xx^{\star} =  (\AA^{\top}\AA)^{-1}\left(\AA^{\top}\bb+\CC^{\top}\left(\CC(\AA^{\top}\AA)^{-1}\CC^{\top}\right)^{-1} \left(\dd - \CC(\AA^{\top}\AA)^{-1}\AA^{\top}\bb\right) \right).
\]

\subsection{Solving \eqref{eq:ProblemSolve}}
At every iteration of the algorithm, we want to solve the following problem,
\begin{align*}
\min_{\Delta}  \quad & \Delta^{\top}\AA^{\top}(\RR + s\II)\AA\Delta\\
& \gg^{\top}\AA\Delta = i/2\\
& \CC\Delta = 0.
\end{align*}
The constraints can be combined and rewritten as, $\CC'\Delta = \dd'$ where,
\[
\CC' = \begin{bmatrix}
	\CC \\
	\gg^{\top}\AA
    \end{bmatrix}, \quad \dd' = \begin{bmatrix}
	0 \\
	i/2
    \end{bmatrix}.
\]
Let $\RR' = \RR+s\II$. We now want to solve, 
\begin{align*}
\min_{\Delta}  \quad & \|{\RR'}^{1/2}\AA\Delta\|_2^2\\
& \CC'\Delta = \dd'.
\end{align*}
Using a procedure similar as in the previous section, we get,
\[
\Dopt = (\AA^{\top}\RR'\AA)^{-1} {\CC'}^{\top}\left(\CC'(\AA^{\top}\RR'\AA)^{-1}{\CC'}^{\top}\right)^{-1}\dd'
\]

\end{appendices}

\end{document}